\documentclass[journal]{IEEEtran}

\makeatletter
\def\endthebibliography{%
	\def\@noitemerr{\@latex@warning{Empty `thebibliography' environment}}%
	\endlist
}
\makeatother

\usepackage{cite}

\usepackage[pdftex]{graphicx} 

\usepackage[caption=false,font=footnotesize]{subfig}
\usepackage[export]{adjustbox}
  
\usepackage{amsmath}
\usepackage{amsfonts}
\usepackage{bm}

\usepackage{filecontents}
\usepackage{amssymb}
\usepackage{color}

\usepackage{epsfig}

\usepackage{verbatim}

\usepackage{url}

\usepackage{algorithm, tabularx}

\usepackage{lettrine}

\usepackage{lipsum}

\usepackage{siunitx}

\usepackage{soul,color}

\usepackage{mathrsfs}

\usepackage{array}

\usepackage[inline]{enumitem}

\usepackage{dblfloatfix} 

\usepackage{algorithm}

\usepackage{algpseudocode}

\makeatletter
\DeclareRobustCommand{\iscircle}{\mathord{\mathpalette\is@circle\relax}}
\newcommand\is@circle[2]{%
  \begingroup
  \sbox\z@{\raisebox{\depth}{$\m@th#1\bigcirc$}}%
  \sbox\tw@{$#1\square$}%
  \resizebox{!}{\ht\tw@}{\usebox{\z@}}%
  \endgroup
}
\makeatother

\usepackage{hhline}

\usepackage{setspace}
\usepackage{multirow}
\usepackage{array}
\newcolumntype{L}[1]{>{\raggedright\let\newline\\\arraybackslash\hspace{0pt}}m{#1}}
\newcolumntype{C}[1]{>{\centering\let\newline\\\arraybackslash\hspace{0pt}}m{#1}}
\newcolumntype{R}[1]{>{\raggedleft\let\newline\\\arraybackslash\hspace{0pt}}m{#1}}

\DeclareMathOperator{\maxdot}{max.}
\DeclareMathOperator{\mindot}{min.}

\newcommand{\ie}{\textit{i}.\textit{e}.\,}
\newcommand{\eg}{\textit{e}.\textit{g}.\,}

\begin{document}
	
\title{Geometry-Aware DRL for Multi-Subband Scheduling in Satellite-Assisted UAM Networks}




	\author{Hyung-Joo Moon,~\IEEEmembership{Member,~IEEE}, Sangha Park,~\IEEEmembership{Graduate Student Member,~IEEE},\\ Chan-Byoung Chae,~\IEEEmembership{Fellow,~IEEE}, and Robert W. Heath, Jr.,~\IEEEmembership{Fellow,~IEEE}

		\thanks{H.-J. Moon, S. Park, and C.-B. Chae are with the School of Integrated Technology, Yonsei University, Seoul 03722, South Korea (e-mail: \{moonhj, shp234, cbchae\}@yonsei.ac.kr).}
        \thanks{R. W. Heath, Jr. is with the Department of Electrical and Computer Engineering, University of California San Diego, La Jolla, CA, 92093 USA (email: rwheathjr@ucsd.edu)}
        }

	\maketitle
	
	\begin{abstract}

In this paper, we investigate downlink scheduling for urban air mobility (UAM) in a cooperative space-air-ground integrated network. Multiple ground stations (GSs) employ narrow three-dimensional beams and share spectrum across multiple subbands, while a satellite provides an orthogonal-band service option. Rapidly time-varying geometry and directional interference require joint decisions on base station association, GS subband assignment, and transmit powers. We formulate a finite-horizon mixed discrete-continuous problem that maximizes sum rate while penalizing handovers and GS overload, using only UAM positions and velocities. To address the combinatorial scheduling problem, we propose GeoSetPPO, a geometry-aware set-attention proximal policy optimization (PPO) method that outputs per-UAM discrete association and subband decisions with permutation-invariant representations. Conditioned on each schedule, GS powers are computed by a per-slot successive convex approximation (SCA) module under per-GS power budgets and minimum signal-to-interference-plus-noise ratio (SINR) constraints. To reduce training cost and improve stability, we adopt a two-stage training strategy that transitions reward evaluation from uniform power to SCA-based power allocation. Simulations demonstrate stable convergence, higher returns than multi-layer perceptron (MLP)- and Transformer-based PPO under the considered training setting, and favorable reward and schedule-feasibility performance relative to algorithm-based and distance-based schedulers. In the larger evaluated network, GeoSetPPO also reduces the scheduling latency from $40.84$~ms to $2.90$~ms relative to the previous algorithm-based method.

	\end{abstract}


	\IEEEpeerreviewmaketitle

\section{Introduction}

\IEEEPARstart{U}{rban} air mobility (UAM) is emerging as a key component of the rapidly expanding low-altitude economy (LAE), enabling new forms of passenger transport, logistics, inspection, and public-safety operations~\cite{laeref}. A critical requirement for large-scale UAM deployment is robust wireless connectivity that supports command-and-control, navigation support, and data services with stringent reliability and latency targets~\cite{202308commag}. Unlike conventional terrestrial users, UAM vehicles move in three-dimensional (3D) space, exhibit rapidly changing geometry relative to the infrastructure, and can experience interference patterns that evolve quickly as aerial users move into or out of each other's angular vicinity. These characteristics make UAM communications inherently mobility-coupled and strongly dependent on spatial configuration.

Motivated by this setting, we develop a centralized downlink architecture where multiple ground stations (GSs) cooperate under a network controller that collects mobility information and coordinates scheduling decisions. The GS-to-UAM channels exhibit strong geometry dependence, where link quality and interference are largely determined by the relative positions and angles between GSs and aerial users~\cite{201903tcom}. This property enables mobility information such as position and velocity to capture the underlying channel conditions, reducing reliance on explicit CSI feedback. At the same time, narrow beams create strongly structured interference: UAMs that are close in angle may suffer significant inter-beam interference, and inter-GS interference can become severe when neighboring GSs steer beams toward nearby aerial directions. As mobility continuously reshapes these angular relations, link association, spectrum usage, and power allocation must be coordinated over time in a rapidly varying network~\cite{myjsac,202402twc,201908iotj}.

To enhance interference management and service robustness, we adopt a cooperative space-air-ground integrated network (SAGIN) in which a satellite operates on a separate frequency band and can serve a subset of UAMs in parallel with the GS tier~\cite{2022icc}. Within this architecture, the prior work~\cite{myjsac} developed a graph-based association method for a single-band GS tier. Although that method accounts for predicted UAM mobility, it selects one association that remains fixed over the considered prediction interval. In contrast, the present formulation permits the association and GS subband assignment to be updated over multiple time steps while explicitly accounting for handovers and subband changes. The resulting decisions are temporally coupled, and the subband assignments determine the interference relationships among GS-served UAMs at each step. These sequential and multi-subband interactions substantially increase the scheduling complexity and cannot be represented directly by the graph-based method in~\cite{myjsac}. We therefore develop a DRL-based framework for joint association and subband scheduling over the considered mobility horizon.

We address this joint design through a hybrid learning-and-optimization approach that reflects the mixed discrete-continuous structure of the problem. We adopt proximal policy optimization (PPO) to learn the sequential association and subband decisions~\cite{ppopaper,202009tcyber}. However, effective learning requires an architecture that represents the set-structured and geometry-dependent interactions among UAMs and communication resources. A conventional multilayer perceptron (MLP) does not explicitly preserve this set structure, while generic Transformer attention lacks the task-specific decomposition of UAM-to-UAM and UAM-to-resource relations that determines directional interference in the considered network. GeoSetPPO incorporates these relations through dedicated geometry-aware set-attention modules and a shared per-UAM decision head. Conditioned on the resulting discrete schedule, the GS power-allocation problem is solved independently at each time step using successive convex approximation (SCA)~\cite{myfsotraj}.

The main contributions of this work are summarized as follows:
\begin{itemize}

    \item We formulate a cooperative ground-satellite downlink problem for mobility-aware UAM communications with a multi-subband GS tier. Unlike~\cite{myjsac}, which maintains one association over a prediction interval and assumes a single shared GS band, the proposed formulation optimizes sequential association and subband decisions while accounting for geometry-driven interference, handovers, GS overload, and per-step power allocation.

    \item We propose GeoSetPPO, a PPO-based scheduler designed for the set-structured UAM scheduling problem. The policy employs dedicated UAM-to-UAM and UAM-to-resource attention modules that incorporate relative position, relative velocity, and resource-related information. This geometry-aware decomposition provides a problem-specific inductive bias that is absent from the considered MLP- and generic Transformer-based PPO architectures.

    \item We address the mixed discrete-continuous design by learning only the discrete association and subband decisions using DRL, while computing continuous GS transmit powers using an SCA-based per-slot power allocation module conditioned on the schedule. We further propose a two-stage training strategy that evaluates rewards with uniform power in the first stage and gradually incorporates SCA-based power allocation, improving training stability and reducing computational cost.
    
    \item We evaluate GeoSetPPO across multiple network scales and subband settings. Under the considered training setting, GeoSetPPO achieves higher returns than the MLP-PPO and Transformer-PPO baselines. Relative to the algorithm-based scheduler~\cite{myjsac} and the distance-based heuristic, it provides favorable reward performance and produces more feasible schedules in the considered multi-subband settings.
\end{itemize}

\subsection{Prior Work}

UAV communication studies broadly consider aerial platforms either as communication infrastructure, such as aerial base stations and relays, or as cellular-connected aerial users. This work belongs to the latter category and focuses on downlink service for UAM vehicles. The air-ground channels of aerial users are strongly geometry-dependent: the high probability of line-of-sight (LoS) propagation and the relatively low path-loss exponent allow both desired and interfering signals to remain significant over long distances, often producing network-wide inter-cell interference~\cite{201902wcom,201907twc,202311tvt,201908iotj}. Multi-antenna transmission, three-dimensional (3D) beamforming, and cooperative interference management have therefore been investigated, together with beam tracking to sustain their gains under mobility~\cite{201912commag,202002tcom,202005wcl}. These physical-layer characteristics motivate network-level coordination of association, spectrum usage, and transmit power.

At the network layer, this design is closely related to joint user association and radio-resource allocation in heterogeneous networks (HetNets), where discrete association and frequency-resource decisions are coupled with interference-dependent power variables, yielding mixed-integer nonconvex formulations. Prior studies have addressed joint association, subchannel, and power allocation using graph-theoretic and difference-of-convex optimization, as well as QoS-aware formulations based on tractable convex approximations and bounds~\cite{wang2017joint,sokun2017qos}. Although such optimization methods can provide high-quality solutions, repeatedly solving them becomes computationally demanding as the numbers of users and resources increase or the network state evolves. This limitation has motivated multi-agent DRL methods for long-term association, spectrum, and power allocation in large-scale HetNets~\cite{zhao2019drl,yang2022distributed}. These studies avoid directly re-solving the original mixed-integer program for every network realization, but focus on terrestrial HetNets without 3D aerial mobility, geometry-structured directional interference, or a satellite service tier.

Joint trajectory and resource allocation has also been studied in aerial and integrated networks. Optimization-based studies have jointly designed UAV placement or trajectories with association, spectrum, and power allocation~\cite{pervez2022joint,yi2024positioning}, while learning-based studies have employed multi-agent DRL and soft actor-critic methods for related joint trajectory-resource decisions~\cite{yin2024joint,202311twc}. Other DRL studies have addressed coverage continuity, spatiotemporal scheduling, and cooperative resource allocation in mobile UAV networks~\cite{202402twc,202305tcom,202003ojvt,2025eatwc,202002twc,202206tcom}. Unlike the present work, however, many of these formulations treat UAVs as controllable aerial BSs and optimize their positions or trajectories, whereas the UAMs considered here are mobile aerial users whose trajectories are externally determined. Related studies have applied attention-based PPO to satellite-assisted computation and MARL to UAM transportation operations, but not to geometry-aware joint association and multi-subband downlink scheduling~\cite{202405tmc,202308tiv}. Our closest prior work developed a graph-theoretic cooperative ground-satellite UAM scheduler~\cite{myjsac}, but considered snapshot-based operation over a single GS band rather than a learned mobility-aware multi-subband policy.

Consequently, the joint design of cooperative GS-satellite association, multi-subband scheduling, QoS-constrained power allocation, and handover-aware management of geometry-structured interference remains underexplored for many moving UAM users. Moreover, many learning-based radio-resource management methods use instantaneous or estimated CSI and link-gain observations as policy inputs, whose frequent acquisition and processing can be difficult in high-mobility air-space networks~\cite{202305tcom,2025eatwc,202009tccn}. These gaps motivate GeoSetPPO, which learns per-UAM association and subband decisions from positions, velocities, and geometry-derived relations through a permutation-invariant set-attention policy, while an SCA module computes continuous GS transmit powers for the selected schedule.

\subsection{Notation and Paper Organization}

Boldface lowercase letters $\bold{x}$ denote vectors, and boldface uppercase letters $\bold{X}$ denote matrices. The operators $(\cdot)^*$, $(\cdot)^\text{T}$, and $(\cdot)^H$ denote complex conjugate, transpose, and Hermitian transpose, and $\otimes$ denotes the Kronecker product. The notation $x\in[a,b]$ indicates that a real number $x$ lies in the closed interval $[a,b]$. The operators $|\cdot|$ and $\lVert\cdot\rVert$ denote absolute value and Euclidean norm, and $\max\{\cdot\}$ denotes the maximum of its arguments. The expectation and variance operators are denoted by $\mathbb{E}[\cdot]$ and $\mathrm{Var}(\cdot)$. The notation $\mathcal{CN}(\mu,\sigma^2)$ denotes the complex normal distribution with mean $\mu$ and variance $\sigma^2$. For simplicity, $\log(\cdot)$ denotes $\log_2(\cdot)$. The function $\mathrm{clip}(x,a,b)$ clips a scalar $x$ to the interval $[a,b]$. Finally, $[\bold{x};\bold{y}]$ denotes vector concatenation.

The remainder of this paper is organized as follows. Section~\ref{sec_sys} presents the cooperative ground-satellite system model, including the geometry, mobility, and link models. Section~\ref{sec_prob} formulates the sequential scheduling and GS power-allocation problems. Section~\ref{sec_proposed} develops the SCA-based power allocator, GeoSetPPO architecture, and two-stage training method. Section~\ref{sec_simres} presents the simulation results, and Section~\ref{sec_conc} concludes the paper.

\section{System Model}
\label{sec_sys}

We investigate a downlink scenario where $K$ GSs and a single satellite jointly serve $M$ UAMs, as illustrated in Fig.~\ref{f01}. The design objective is to maximize the network sum rate while penalizing excessive handovers and overly unbalanced user association across GSs. We refer to the GSs and the satellite collectively as base stations (BSs). The GS tier operates over multiple subbands, and the central task at each time step is to determine link association, subband assignment, and GS transmit powers for a given spatial configuration of GSs and UAMs. In our setting, association decisions are driven by beam directions and 3D geometry (\eg, positions and velocities), together with network-level constraints such as per-BS user capacity.

We also incorporate a satellite tier operating on a dedicated frequency band and serving a subset of UAMs under a centralized controller that performs joint scheduling across the GS and satellite tiers. We consider a low Earth orbit (LEO) satellite forming a quasi-stationary spot beam that covers the target airspace. Since the satellite footprint typically spans several tens of kilometers, we do not consider per-UAM satellite beamforming. Instead, satellite-served users are multiplexed via time-division multiple access (TDMA) or frequency-division multiple access (FDMA). The satellite tier therefore provides an additional orthogonal-band downlink option that complements the GS tier depending on the instantaneous geometry and interference conditions. We assume perfect time and frequency synchronization between each connected BS-UAM pair via robust signaling. The associated synchronization signaling overhead incurred when a UAM attaches to a new BS is captured by the handover penalty.

\begin{figure}[t]
	\begin{center}
		{\includegraphics[width=0.85\columnwidth,keepaspectratio]
			{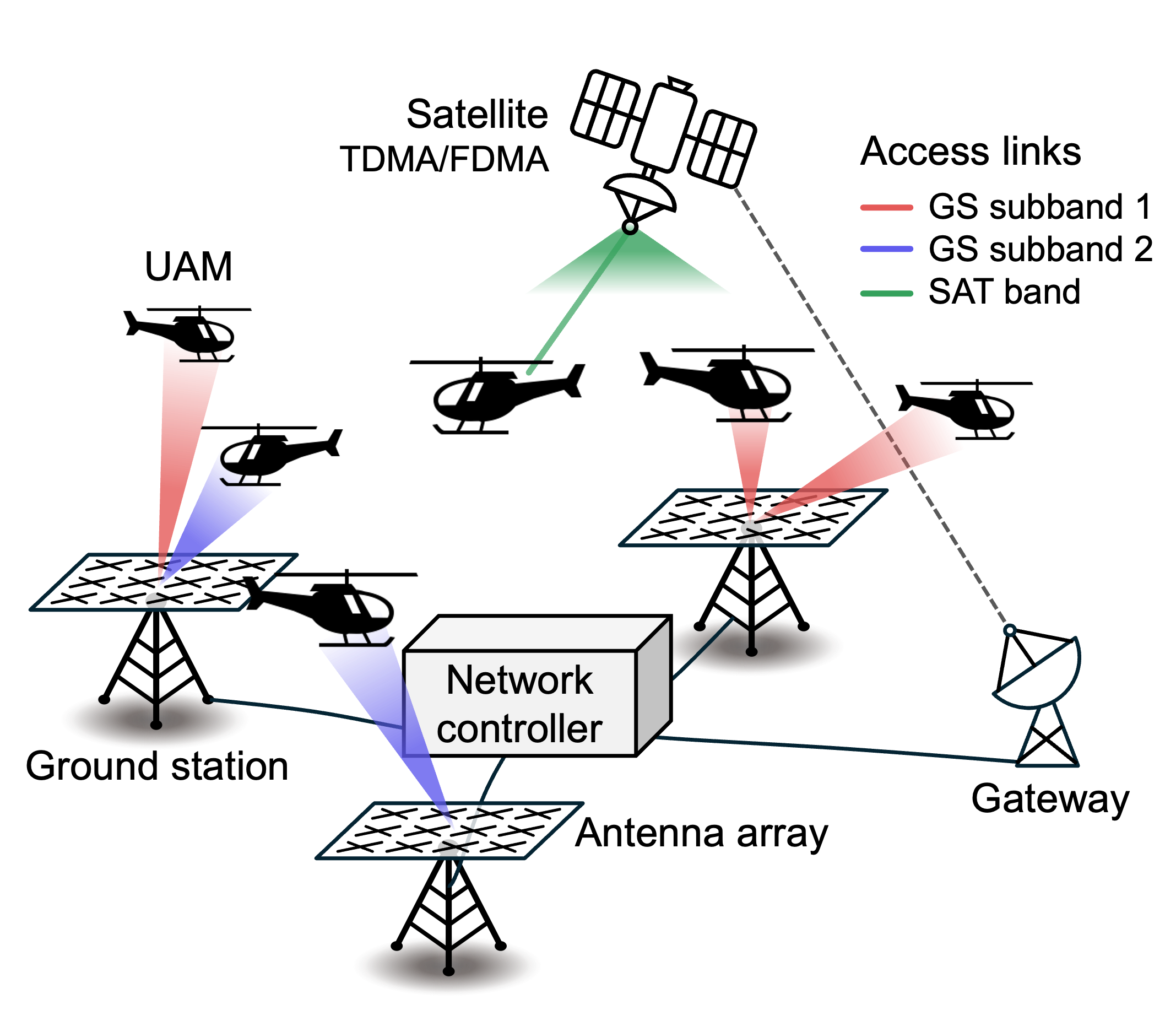}%
			\caption{System model of a space-air-ground integrated UAM network.}
			\label{f01}
		}
	\end{center}
	\vspace{-10pt}
\end{figure}

\subsection{Geometry-Related Parameters}

Each GS is equipped with an upward-facing rectangular antenna array of dimension $N_x \times N_y$, where $N_x$ and $N_y$ denote the numbers of elements along the horizontal and vertical axes, respectively. The associated unit vectors $\hat{\bold{i}}^{(k)},\,\hat{\bold{j}}^{(k)}\in\mathbb{R}^3$ describe the local array orientation of the $k$-th GS array (the x- and y-axis directions). To provide sufficiently narrow beams for UAMs at distances of several hundred meters or more, both $N_x$ and $N_y$ are chosen to be moderately large (\eg, $N_x,N_y\geq 4$). Each antenna element is connected to its own RF chain, enabling fully digital beamforming and simultaneous service of multiple UAMs per GS. To avoid excessive complexity and backhaul congestion, however, the DRL scheduler imposes a penalty whenever the number of UAMs associated with a GS exceeds a threshold $\mathcal{N}_\text{GS}$. After link association, we denote the number of GS-served UAMs by $M_\text{GS}$ and the number of satellite-served UAMs by $M_\text{SAT}$, with $M_\text{GS} + M_\text{SAT} = M$. Each UAM is assumed to carry two antennas, a downward-facing antenna for GS links and a directional upward antenna for satellite connectivity.

We model each mobility episode over $T_\text{s}$ discrete time steps indexed by $t\in\{0,\cdots,T_\text{s}-1\}$, with sampling interval $\tau_\text{s}$ between consecutive steps. The satellite position is defined as
\begin{equation}
\label{satloc_new}
\bold{u}^\text{S}[t] = \big[x^\text{S}[t],\,y^\text{S}[t],\,z^\text{S}\big]^\text{T},
\end{equation}
where $z^\text{S}$ is the orbital altitude. The coordinates of the $k$-th GS are
\begin{equation}
\label{gsloc_new}
\bold{u}_k^\text{G} = \big[x_k^\text{G},\,y_k^\text{G},\,z_k^\text{G}\big]^\text{T},
\end{equation}
for $k\in\{1,\cdots,K\}$. The trajectory of the $m$-th UAM is described by its position and velocity:
\begin{equation}
\label{uamloc_new}
\begin{aligned}
\bold{u}_m[t] &= \big[x_m^{\text{U}}[t],\,y_m^{\text{U}}[t],\,z_m^{\text{U}}[t]\big]^\text{T},\\
\dot{\bold{u}}_m[t] &= \big[\dot{x}_m^{\text{U}}[t],\,\dot{y}_m^{\text{U}}[t],\,\dot{z}_m^{\text{U}}[t]\big]^\text{T},
\end{aligned}
\end{equation}
with $m\in\{1,\cdots,M\}$ indicating the UAM index. In the simulations, we generate smooth UAM trajectories by controlling the acceleration profiles. For any time index $t$ when scheduling is performed, we assume that the network controller knows $\bold{u}_k^\text{G}$ for all GSs and $\bold{u}_m[t],\dot{\bold{u}}_m[t]$ for all UAMs. Specifically, the displacement of UAM $m$ from GS $k$ is
\begin{equation}
\label{dvec_new}
\bold{d}_m^{(k)}[t] = \bold{u}_m[t] - \bold{u}_k^\text{G}.
\end{equation}
The cosines between $\bold{d}_m^{(k)}[t]$ and the array axes $\hat{\bold{i}}^{(k)}$ and $\hat{\bold{j}}^{(k)}$ are represented as
\begin{equation}
\label{anglealphabeta_new}
\begin{aligned}
\cos\alpha_{m}^{(k)}[t]&=\frac{\bold{d}_m^{(k)}[t]^\text{T}\hat{\bold{i}}^{(k)}}{\lVert\bold{d}_m^{(k)}[t]\rVert},\\
\cos\beta_{m}^{(k)}[t]&=\frac{\bold{d}_m^{(k)}[t]^\text{T}\hat{\bold{j}}^{(k)}}{\lVert\bold{d}_m^{(k)}[t]\rVert}.
\end{aligned}
\end{equation}
Thus, the angles $\alpha_{m}^{(k)}[t]$ and $\beta_{m}^{(k)}[t]$ define the link direction with respect to GS $k$.

\subsection{GS-to-UAM Signal Models}

To model inter-beam and inter-cell interference under simultaneous GS transmissions, we relabel the UAM indices after link scheduling. While the UAMs were indexed by $m$ in the previous geometry description, after scheduling we index the GS-served UAMs by a triple $(k,b,p)$, or simply $kbp$, where $k\in\{1,\cdots,K\}$ is the GS index, $b\in\{1,\cdots,B\}$ is the GS subband index, and $p\in\{1,\cdots,P_{kb}\}$ is the local UAM index associated with GS $k$ on subband $b$. Here, $P_{kb}$ denotes the number of UAMs scheduled by GS $k$ on subband $b$, so that $\sum_{k=1}^{K}\sum_{b=1}^{B} P_{kb} = M_\text{GS}$. By construction, the UAM labeled by $kbp$ is served by GS $k$ on subband $b$. Let $W_\mathrm{G}$ denote the total bandwidth of the GS tier. We assume that $W_\mathrm{G}$ is equally partitioned into $B$ orthogonal subbands, so that each subband has bandwidth $W_\mathrm{G}/B$.

We now describe the baseband-equivalent signal and channel models for the downlink GS-to-UAM transmission. For notational simplicity, although most quantities depend on the time-varying locations of the UAMs, we omit the explicit time index $[t]$ in this subsection. Let $x_{kbp}\in\mathbb{C}$ be the data symbol for UAM $kbp$, and let $\rho_{kbp}\in\mathbb{R}$ be the transmit power allocated to that stream. The 3D beamforming vector at GS $k$ for UAM $kbp$ is denoted by $\bold{v}_{kbp}^{(k)}\in\mathbb{C}^{N_x N_y}$. Then, the transmit signal vector of GS $k$ on subband $b$ is modeled as
\begin{equation}
\label{basetx}
\bold{y}_{kb} = \sum_{p=1}^{P_{kb}}\bold{v}_{kbp}^{(k)} \sqrt{\rho_{kbp}}\, x_{kbp}.
\end{equation}
To construct $\bold{v}_{kbp}^{(k)}$, we first define the Vandermonde vector as
\begin{equation}
\label{arrresp}
\bold{a}_N(x) = \big[1,\,e^{-j\pi x},\,\cdots,\,e^{-j\pi (N-1)x}\big]^\text{T}.
\end{equation}
We assume that the GS-to-UAM links are LoS-dominant. Using the angles between the $k$-th GS array and the link direction, $(\alpha_{kbp}^{(k)},\beta_{kbp}^{(k)})$ from~\eqref{anglealphabeta_new}, a normalized 3D maximum ratio transmission (MRT) beamforming vector is designed as
\begin{equation}
\label{beamvec}
\bold{v}_{kbp}^{(k)}
= \frac{1}{\sqrt{N_x N_y}}
\big\{\bold{a}_{N_x}(\cos\alpha_{kbp}^{(k)}) \otimes
      \bold{a}_{N_y}(\cos\beta_{kbp}^{(k)})\big\}^*.
\end{equation}
Let $L_{kbp}^{(\ell)}\in\mathbb{R}$ denote the free-space pathloss (FSPL) between GS $\ell$ and UAM $kbp$. Let $N_0$ denote the thermal-noise power spectral density. Since each GS subband has bandwidth $W_\mathrm{G}/B$, the
receiver-noise power on one subband is
\begin{equation}
\label{gsnoisepower}
\sigma_{\mathrm{n},B}^{2}
=
N_0\frac{W_\mathrm{G}}{B}.
\end{equation}
Accordingly, let $z_{kbp}\sim\mathcal{CN}(0,\sigma_{\mathrm{n},B}^{2})$ denote the receiver noise, and let $g_{kbp}^{\ell bq}\in\mathbb{C}$ denote the normalized beamforming gain (\ie, the inner product between the channel vector and the beamforming vector) observed at UAM $kbp$ on subband $b$ when GS $\ell$ transmits using the beam intended for UAM $\ell bq$. The received baseband signal at UAM $kbp$ on subband $b$ is then
\begin{equation} 
\label{rxsignal0}
\begin{aligned}
r_{kbp}
&= \sqrt{L_{kbp}^{(k)}}\, g_{kbp}^{kbp}\, \sqrt{\rho_{kbp}}\, x_{kbp}\\
&\quad + \sum_{\ell\neq k}^{K} \sum_{q=1}^{P_{\ell b}}
\sqrt{L_{kbp}^{(\ell)}}\, g_{kbp}^{\ell bq}\, \sqrt{\rho_{\ell bq}}\, x_{\ell bq}\\
&\quad + \sum_{q\neq p}^{P_{kb}}
\sqrt{L_{kbp}^{(k)}}\, g_{kbp}^{kbq}\, \sqrt{\rho_{kbq}}\, x_{kbq}\,+ z_{kbp},
\end{aligned}
\end{equation}
where the three terms correspond to the desired signal, inter-GS interference, and intra-GS interference, respectively.
Assuming isotropic GS-link antenna elements and neglecting backlobe radiation, let $\lambda$ denote the carrier wavelength and let $\bold{d}_{kbp}^{(\ell)}$ denote the displacement vector from GS $\ell$ to UAM $kbp$. The resulting FSPL is
\begin{equation}
\label{fspl}
L_{kbp}^{(\ell)} = \bigg(\frac{\lambda}{4\pi \lVert \bold{d}_{kbp}^{(\ell)} \rVert}\bigg)^2.
\end{equation}
We use the geometry-derived LoS array response to design the GS beams and construct the scheduling inputs. The GeoSetPPO actor therefore does not observe instantaneous small-scale fading realizations. During training and evaluation, however, the link gains used to compute the scheduling reward follow a LoS-dominant Rician model with $K_{\mathrm R}=20$~dB unless otherwise stated. Section~\ref{sec_simres} further evaluates policies trained at this nominal condition under different test-channel Rician factors. Accordingly, the deterministic LoS component is modeled as
\begin{equation}
\label{channelvec}
\bold{h}_{kbp}^{(\ell)} = \bold{a}_{N_x}(\cos\alpha_{kbp}^{(\ell)}) \otimes
\bold{a}_{N_y}(\cos\beta_{kbp}^{(\ell)}).
\end{equation}
The corresponding beamforming gain when GS $\ell$ uses the beam for UAM $\ell bq$ and UAM $kbp$ receives is
\begin{equation}
\label{bfgain}
\begin{aligned}
g_{kbp}^{\ell bq}
&= \bold{h}_{kbp}^{(\ell)\,\text{T}} \bold{v}_{\ell bq}^{(\ell)} \\
&= \frac{1}{\sqrt{N_x N_y}}
\bold{a}_{N_x}(\cos\alpha_{kbp}^{(\ell)})^\text{H}
\bold{a}_{N_x}(\cos\alpha_{\ell bq}^{(\ell)}) \\
&\quad\cdot
\bold{a}_{N_y}(\cos\beta_{kbp}^{(\ell)})^\text{H}
\bold{a}_{N_y}(\cos\beta_{\ell bq}^{(\ell)}).
\end{aligned}
\end{equation}
Thus, $g_{kbp}^{kbp}$ captures the main-lobe gain for the intended UAM, while $g_{kbp}^{\ell bq}$ with $\ell\neq k$ or $q\neq p$ represents inter-beam interference.

Both $L_{kbp}^{(\ell)}$ and $g_{kbp}^{\ell bq}$ are determined by the geometry and the array configurations and can be precomputed. We therefore define the combined large-scale channel coefficient
\begin{equation}
\label{winter}
w_{kbp}^{\ell bq} = L_{kbp}^{(\ell)} \big|g_{kbp}^{\ell bq}\big|^2.
\end{equation}
In particular, we denote the desired-link coefficient by $w_{kbp} \triangleq w_{kbp}^{kbp}$. Using the channel coefficients $w_{kbp}^{\ell bq}$ and the GS transmit power $\rho_{kbp}$, the signal-to-interference-plus-noise ratio (SINR) at UAM $kbp$ on subband $b$ can be expressed as
\begin{equation}
\label{sinr0}
\gamma_{kbp} =
\frac{w_{kbp}\,\rho_{kbp}}{
\displaystyle \sum_{\ell\neq k}^{K} \sum_{q=1}^{P_{\ell b}} w_{kbp}^{\ell bq}\, \rho_{\ell bq}
+ \sum_{q\neq p}^{P_{kb}} w_{kbp}^{kbq}\, \rho_{kbq}
+ \sigma_{\mathrm{n},B}^{2} }.
\end{equation}
The corresponding achievable spectral efficiency is modeled as follows:
\begin{equation}
\label{capacity0}
C_{kbp} = \log\big(1+\gamma_{kbp}\big).
\end{equation}
Since each GS subband has bandwidth $W_\mathrm{G}/B$, the achievable data rate of UAM $kbp$ is
\begin{equation}
\label{gslinkrate}
R_{kbp}
=
\frac{W_\mathrm{G}}{B}C_{kbp}.
\end{equation}
Accordingly, when the aggregate GS-tier rate is normalized by the fixed total GS bandwidth $W_\mathrm{G}$, it is expressed as
\begin{equation}
\label{gsnormalizedrate}
\overline{R}_\mathrm{G}
=
\frac{1}{B}
\sum_{k=1}^{K}
\sum_{b=1}^{B}
\sum_{p=1}^{P_{kb}}
C_{kbp}.
\end{equation}
Although small-scale fading makes the instantaneous downlink SINR random, the scheduling inputs and SCA-based power allocation are computed using the geometry-derived LoS channel gains. During training and evaluation, the rate expression in~\eqref{capacity0} is evaluated using Rician channel gains only for reward calculation, without exposing the fading realizations to either GeoSetPPO or the SCA module. This is motivated by the fact that, in aerial links, the dominant interference structure is primarily governed by slowly varying geometry and beam directions, whereas instantaneous CSI is difficult to acquire and exploit in high-mobility scenarios. For brevity, we omit the explicit discrete time index $[t]$ in this subsection, however, most parameters defined in this section are functions of $t$ and are used in the subsequent analysis.

\subsection{Assumptions for Satellite-to-UAM Transmission}

After link scheduling, $M_\text{SAT}$ UAMs are assigned to the satellite, which serves them in a separate frequency band from the GS tier. Let $W_\mathrm{S}$ denote the satellite-tier bandwidth. For simplicity, we assume $W_\mathrm{S}=W_\mathrm{G}$ throughout this work. Due to the high altitude of the satellite, the spatial channels of different satellite-to-UAM links are highly correlated. We therefore assume that the satellite downlink operates in a TDMA or FDMA mode. Let $s\in\{1,\cdots,M_\text{SAT}\}$ denote the index of a satellite-served UAM. The corresponding satellite-to-UAM channel gain is modeled as
\begin{equation}
\label{fsplsat_new}
\mathcal{G}_s^\text{S}
= G_\text{S}^\text{R} G_\text{S}^\text{T}(\mu_{s})
\left(\frac{\lambda_\text{S}}{4\pi d_s^\text{S}}\right)^2
\mathcal{H}_{s}^\text{S}\mathcal{L}_\text{S},
\end{equation}
where $G_\text{S}^\text{R}$ and $G_\text{S}^\text{T}(\mu_{s})$ are the receive and transmit antenna gains, $d_s^\text{S}$ is the satellite-to-UAM distance, $\lambda_\text{S}$ is the satellite carrier wavelength, $\mathcal{H}_{s}^\text{S}$ collects propagation loss factors (\eg, gaseous, cloud, and rain attenuation), and $\mathcal{L}_\text{S}$ represents other losses such as antenna losses~\cite{myjsac}.

A widely used K-band tapered-aperture antenna model~\cite{1991jsac} characterizes the satellite transmit gain. Let $J_1(\cdot)$ and $J_3(\cdot)$ denote the first- and third-order Bessel functions, respectively. The transmit gain is given by
\begin{equation}
\label{satbeam_new}
G_\text{S}^\text{T}(\mu_{s})
= G_0\left[\frac{J_1(\mu_{s})}{2\mu_{s}}
+ 36\frac{J_3(\mu_{s})}{\mu_{s}^3}\right]^2.
\end{equation}
Here, the argument $\mu_{s}$ is given by $\mu_{s} = 2.07123\sin(\psi_{s})/\sin(\psi_\text{3dB})$, where $\psi_{s}$ is the off-boresight angle and $\psi_\text{3dB}=0.39{\pi\lambda_\text{S}}/{a}$ is the 3-dB beamwidth for an aperture diameter $a$~\cite{202202twcbeam}. The boresight gain is given by $G_0 = \frac{4\pi A\eta}{\lambda_\text{S}^2}$, where $A$ is the effective aperture and $\eta$ is the antenna efficiency.
    
Additional attenuation factors such as atmospheric gases, clouds, and rain can be included using ITU-R recommendations~\cite{itu676,itu840,itu838}, and are combined into $\mathcal{H}_{s}^\text{S}$. The analysis in~\cite{myjsac} showed that the overall expected satellite-to-UAM channel gains for different UAMs in the target airspace differ by less than $0.5$~dB, even under distinct weather conditions. Motivated by this, we assume that all satellite-served UAMs experience approximately the same downlink channel gain. Hence, the satellite downlink signal-to-noise ratio (SNR) can be approximated by the common expression
\begin{equation}
\label{satsinr}
\gamma^\text{SAT}
= \frac{G_\text{S}^\text{R} G_\text{S}^\text{T}(0)
\left(\frac{\lambda_\text{S}}{4\pi z^\text{S}}\right)^2
\mathcal{H}^\text{S}\mathcal{L}_\text{S}\,\rho_{\text{S}}}{\sigma_\text{S}^2},
\end{equation}
where $z^\text{S}$ is the satellite altitude, $\rho_{\text{S}}$ is the satellite transmit power, $\mathcal{H}^\text{S}$ is the common propagation-loss factor, and $\sigma_\text{S}^2$ is the noise power at the UAM receiver. The corresponding achievable spectral efficiency on the satellite link is
\begin{equation}
\label{satcapacity}
C_\text{SAT} = \log\big(1+\gamma^\text{SAT}\big).
\end{equation}
Since a time- or frequency-division multiple access scheme is used, the $M_\text{SAT}$ UAMs orthogonally share this spectral efficiency. Let $C_s^\text{SAT}$ denote the effective satellite downlink spectral efficiency of the $s$-th satellite-served UAM. The sum of the individual spectral efficiencies then satisfies
\begin{equation}
\label{satsumrate}
\sum_{s=1}^{M_\text{SAT}} C_s^\text{SAT} = C_\text{SAT}.
\end{equation}
Thus, assigning additional UAMs to the satellite redistributes the fixed satellite-tier spectral efficiency among its users.

\section{Problem Formulation}
\label{sec_prob}

Our objective is to design a joint link-association, subband-allocation, and GS power-allocation strategy that maximizes the cumulative sum rate over the considered horizon of all UAMs. These decisions are made by a centralized network controller using only the current UAM positions and velocities, without requiring instantaneous CSI feedback from UAMs to the BSs. At each decision step $t$, the controller selects the association, subband assignment, and transmit powers that remain fixed over the actual time interval $\tau\in[t\tau_s,(t+1)\tau_s)$. The index $t$ is also used as the time step in the DRL formulation. At each decision step, we decompose the overall problem into two coupled components:
\begin{itemize}
    \item \textbf{Scheduling:} BS-UAM association (GS tier or satellite tier) and subband selection for GS-UAM links.
    \item \textbf{GS power allocation:} transmit-power optimization for the scheduled GS-UAM links subject to per-GS power budgets and per-link minimum-SINR constraints.
\end{itemize}
In this section, we first formulate the scheduling problem and then describe its coupling with the GS power allocation subproblem.

\subsection{Scheduling Problem}
\label{schedprob}

At each decision step $t$, the network controller allocates all $M$ UAMs to the available space-frequency resources. Each UAM can either be associated with one of the $K$ GSs on one of the $B$ subbands, or be served by the satellite. Therefore, each UAM has $KB+1$ association options. The scheduling objective is to maximize the sum rate over the entire horizon $t \in \{0,\cdots,T_\text{s}-1\}$, given that the GS powers are optimized by a separate power allocation algorithm, while at the same time reducing handover events and avoiding overloading any single GS.

Let $f_t^\text{sched}$ denote the scheduling function at time step $t$, which maps the original UAM index $m\in\{1,\cdots,M\}$ to either a GS-subband pair or the satellite:
\begin{equation}
\label{schedfunc}
\begin{aligned}
f_t^\text{sched} : \{1,\cdots,M\}
&\to \{(\text{GS}1,\text{band}1),(\text{GS}1,\text{band}2),\cdots,\\
&\qquad (\text{GS}K,\text{band}B),\text{SAT}\}.
\end{aligned}
\end{equation}
According to this scheduling function, each UAM index $m$ is mapped to either $kbp$ (GS-served UAM index) or $s$ (satellite-served UAM index), as described in Section~\ref{sec_sys}. In the following formulations, indices of the form $kbp$ or $s$ refer to specific UAMs.

Using the rate expressions in~\eqref{capacity0} and~\eqref{satcapacity}, the finite horizon scheduling problem is formulated as
\begin{subequations}
\label{problem1}
\begin{align}
\text{(P1)}&:\notag\\
\begin{split}
\label{po1}
\underset{\{f_t^\text{sched}\}}{\maxdot}&\;\;
\sum_{t=0}^{T_\text{s}-1}
\Bigg\{
\frac{1}{B}
\sum_{k=1}^{K}\sum_{b=1}^{B}\sum_{p=1}^{P_{kb}[t]}
C_{kbp}[t] \\
&\qquad\qquad\qquad
+\mathcal{R}_t^\text{sat}
-\mathcal{P}_t^\text{ho}
-\mathcal{P}_t^\text{gs}
\Bigg\}
\end{split}
\\
\begin{split}
\label{pc11}
\textrm{s.t.}&\quad
\mathcal{R}_t^\text{sat}
=\begin{cases}
0, & M_\text{SAT}[t]=0\\
C_\text{SAT}, & M_\text{SAT}[t]\geq 1
\end{cases}
\end{split}
\\
\begin{split}
\label{pc12}
&\quad\mathcal{P}_t^\text{ho}
= c^\text{B}\mathcal{N}_t^\text{B}
+ c^\text{G}\mathcal{N}_t^\text{G}
+ c^\text{S}\mathcal{N}_t^\text{S}
\end{split}
\\
\begin{split}
\label{pc13}
&\quad\mathcal{P}_t^\text{gs}= c_\text{gs}
\sum_{k=1}^{K}
\max\big\{M_\text{GS}^{(k)}[t]-\mathcal{N}_\text{GS},\,0\big\}
\end{split}
\\
\begin{split}
\label{pc14}
&\quad\forall(k,b,p,t):\\
&\quad\rho_{kbp}[t]\text{ are determined by solving~(P2)}
\end{split}
\end{align}
\end{subequations}
In~(P1), the first term in~\eqref{po1} is the aggregate GS-tier rate normalized by the fixed total GS bandwidth $W_\mathrm{G}$, where $C_{kbp}[t]$ is computed from~\eqref{capacity0}. For a given scheduling decision, the remaining degrees of freedom in $C_{kbp}[t]$ are the transmit powers $\rho_{kbp}[t]$, which are obtained by solving the GS power allocation subproblem~(P2), as indicated in~\eqref{pc14} and detailed in Section~\ref{gspowerprob}. The term $\mathcal{R}_t^\text{sat}$ denotes the satellite-tier throughput. Constraint~\eqref{pc11} specifies that the satellite tier provides zero throughput when no UAM is assigned to the satellite, and a fixed sum throughput $C_\text{SAT}$ otherwise, consistent with~\eqref{satsumrate}. For simplicity, $C_\text{SAT}$ is treated as time-invariant. The handover penalty $\mathcal{P}_t^\text{ho}$ in~\eqref{pc12} assigns different costs through the coefficients $c^\text{B}$, $c^\text{G}$, and $c^\text{S}$ to three types of switching events: $\mathcal{N}_t^\text{B}$ counts subband changes for UAMs that remain associated with the same GS, $\mathcal{N}_t^\text{G}$ counts GS-to-GS handovers between consecutive slots, and $\mathcal{N}_t^\text{S}$ counts switching events between the GS tier and the satellite tier. Finally, $\mathcal{P}_t^\text{gs}$ in~\eqref{pc13} discourages GS overload: if the number of UAMs associated with GS $k$ at time $t$, denoted by $M_\text{GS}^{(k)}[t]$, exceeds the threshold $\mathcal{N}_\text{GS}$, a penalty proportional to $c_\text{gs}$ is imposed.

Problem (P1) is a mixed-integer program with a dual-loop structure. The outer loop selects the discrete association and subband decisions, while the inner loop optimizes the continuous GS transmit powers through (P2) at each time step. Each UAM selects one of $KB$ GS-subband pairs or the satellite tier, resulting in $(KB+1)^M$ possible joint scheduling decisions per step. Directly enumerating the corresponding decision sequences over $T_\mathrm{s}$ steps would require considering $(KB+1)^{MT_\mathrm{s}}$ possibilities. Moreover, the subband assignments determine the same-band interference relationships within each step, while the handover penalties couple successive decisions. Compared with~\cite{myjsac}, which maintains one association over a prediction interval and assumes a single shared GS band, these sequential and multi-subband couplings prevent direct application of the graph-based association procedure and make exhaustive optimization impractical.

\subsection{GS Power Allocation Problem}
\label{gspowerprob}

Given a scheduling decision $f_t^\text{sched}$ at time step $t$, suppose that all $M$ UAMs have been assigned either to a GS-subband pair or to the satellite. Accordingly, GS-served UAMs are indexed by $(k,b,p)$ and the satellite-served UAMs by $s$. The remaining task is to allocate transmit powers to the scheduled GS-UAM links. Since the satellite operates on a separate band with a fixed sum throughput $C_\text{SAT}$, the power allocation affects only the sum rate of GS-served UAMs. For a fixed $t$ (and omitting the explicit time index for notational simplicity), the GS power allocation problem is formulated as
\begin{subequations}
\label{problem2}
\begin{align}
\text{(P2)}:\;\;
\begin{split}
\label{po2}
\underset{\boldsymbol{\rho}}{\mindot}\quad
& -\sum_{k=1}^{K} \sum_{b=1}^{B} \sum_{p=1}^{P_{kb}} C_{kbp}
\end{split}
\\
\begin{split}
\label{pc21}
\textrm{s.t.}\quad
& \sum_{b=1}^{B} \sum_{p=1}^{P_{kb}} \rho_{kbp} \leq \rho_\text{tot},\quad \forall k
\end{split}
\\
\begin{split}
\label{pc22}
& \rho_{kbp} \geq 0,\quad \forall (k,b,p)
\end{split}
\\
\begin{split}
\label{pc23}
& \gamma_{kbp} \geq \gamma_\text{min},\quad \forall (k,b,p)
\end{split}
\end{align}
\end{subequations}
For a fixed value of $B$, the factor $1/B$ is constant with respect to the GS transmit powers. Therefore, it is omitted from the power-allocation subproblem without affecting its optimizer.
Here, $\boldsymbol{\rho}=[\rho_{111},\cdots,\rho_{KBP_{KB}}]^\text{T}$ denotes the collection of GS transmit powers, $\rho_\text{tot}$ is the per-GS sum-power budget, and $\gamma_\text{min}$ is the minimum target SINR for each GS-served UAM. Constraints~\eqref{pc21}-\eqref{pc23} enforce, respectively, the per-GS power budget, non-negativity of transmit powers, and per-link minimum SINR requirements.Thus, whenever (P2) is feasible, the minimum-SINR constraint enforces a link-level QoS requirement for every GS-served UAM.

\begin{figure}[t]
	\begin{center}
		{\includegraphics[width=0.95\columnwidth,keepaspectratio]
			{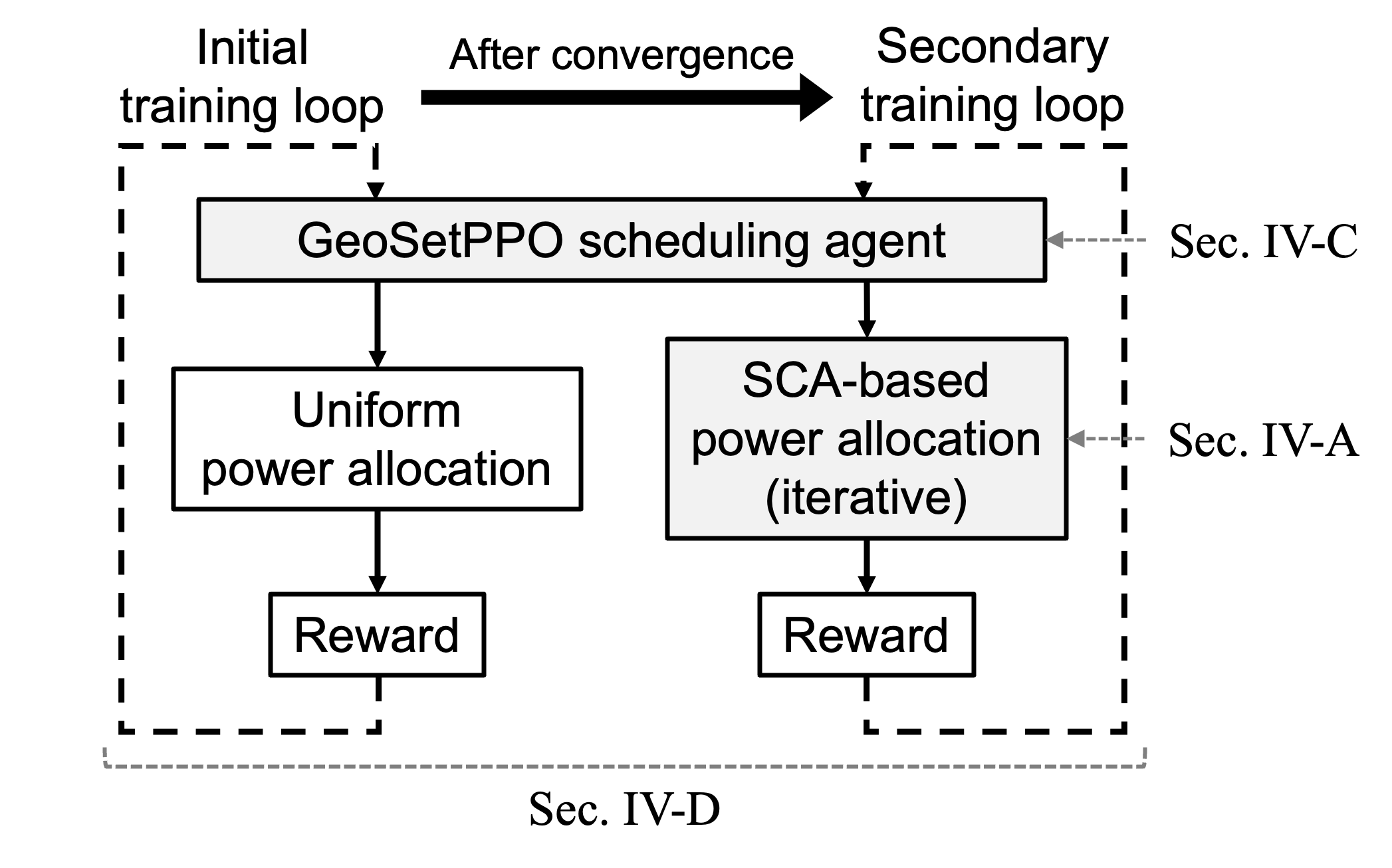}%
			\caption{Proposed two-stage PPO training strategy. The first stage evaluates rewards using uniform GS power allocation, and the second stage evaluates rewards using the SCA-based power allocator.}
			\label{f02}
		}
	\end{center}
	\vspace{-10pt}
\end{figure}

\section{DRL-Based UAM Scheduling Framework}
\label{sec_proposed}

To solve the scheduling problem (P1), we propose a geometry-aware set-attention PPO (GeoSetPPO)-based DRL framework. The agent learns a scheduling policy by maximizing a reward that combines the sum-rate term with the handover and GS-overload penalties. Because the scheduling objective in~\eqref{po1} depends on the transmit powers through~\eqref{pc14}, computing the reward at each time step requires a specific power-allocation procedure. We therefore treat GS power allocation as an inner-loop module: given a schedule, the GS transmit powers are obtained by an SCA-based solver, and the resulting rates $C_{kbp}[t]$ are used to evaluate the reward and update the PPO policy in the outer loop.

In this section, we first develop an SCA-based nonlinear programming method to solve (P2), which provides a practical power-allocation module for the multi-GS, multi-subband setting. We then present the PPO training procedure and the proposed GeoSetPPO actor-critic architecture, which employs set attention to capture UAM-UAM and UAM-GS interactions. Finally, as illustrated in Fig.~\ref{f02}, we introduce a two-stage training strategy that gradually transitions reward evaluation from uniform to SCA-based power allocation, enabling stable and effective learning for (P1).

\subsection{Nonlinear Programming-Based Power Allocation Module}

For a given schedule, (P2) maximizes the GS-tier sum rate subject to per-GS power and per-link minimum SINR constraints. While the constraints in~\eqref{pc21}-\eqref{pc23} are convex, the objective in~\eqref{po2} is non-convex. We therefore adopt SCA, a widely used iterative nonlinear programming technique~\cite{myfsotraj,myjsac}. The key idea is to replace the non-convex objective with a convex upper bound that is tight at the current iterate. At SCA iteration $\delta$, given the current power vector $\boldsymbol{\rho}^{(\delta)}$, we construct a convex surrogate for $-C_{kbp}$ and minimize the sum of these surrogates. The resulting subproblem is convex and can be solved efficiently, and its optimizer is used as the expansion point for the next iteration.

\newtheorem{lemma}{Lemma}
\begin{lemma}
\label{lemma1}
For each scheduled link $(k,b,p)$, the non-convex function $-C_{kbp}$  can be upper-bounded by the following convex surrogate at iteration $(\delta+1)$:
\begin{equation}
\label{approx01}
\begin{aligned}
- C_{kbp}^{(\delta+1)}
&= -\theta_{kbp}^{(\delta)} \log\big(w_{kbp} \rho_{kbp}\big) \\
+\frac{\theta_{kbp}^{(\delta)}}{\mu_{kbp}^{(\delta)}}&\bigg(
\sum_{\ell\neq k}^{K} \sum_{q=1}^{P_{\ell b}} w_{kbp}^{\ell bq}\rho_{\ell bq}
+ \sum_{q\neq p}^{P_{kb}} w_{kbp}^{kbq}\rho_{kbq}
\bigg)
- \zeta_{kbp}^{(\delta)},
\end{aligned}
\end{equation}
where the parameters $\theta_{kbp}^{(\delta)}$, $\mu_{kbp}^{(\delta)}$, and $\zeta_{kbp}^{(\delta)}$ are defined as
\begin{equation}
\label{thetat}
\theta_{kbp}^{(\delta)} =
\frac{w_{kbp} \rho_{kbp}^{(\delta)}}{\sum_{\ell=1}^{K} \sum_{q=1}^{P_{\ell b}} w_{kbp}^{\ell bq} \rho_{\ell bq}^{(\delta)} + \sigma_{\mathrm{n},B}^{2}},
\end{equation}
\begin{equation}
\label{mut}
\mu_{kbp}^{(\delta)} =
\sum_{\ell\neq k}^{K} \sum_{q=1}^{P_{\ell b}} w_{kbp}^{\ell bq} \rho_{\ell bq}^{(\delta)}
+ \sum_{q\neq p}^{P_{kb}} w_{kbp}^{kbq} \rho_{kbq}^{(\delta)}
+ \sigma_{\mathrm{n},B}^{2},
\end{equation}
\begin{equation}
\label{zetat}
\begin{aligned}
\zeta_{kbp}^{(\delta)}
&= \log \bigg(1 + \frac{w_{kbp} \rho_{kbp}^{(\delta)}}{\mu_{kbp}^{(\delta)}}\bigg)
- \theta_{kbp}^{(\delta)} \log\big(w_{kbp} \rho_{kbp}^{(\delta)}\big) \\
&\quad + \frac{\theta_{kbp}^{(\delta)}}{\mu_{kbp}^{(\delta)}}
\bigg(
\sum_{\ell\neq k}^{K} \sum_{q=1}^{P_{\ell b}} w_{kbp}^{\ell bq}\rho_{\ell bq}^{(\delta)}
+ \sum_{q\neq p}^{P_{kb}} w_{kbp}^{kbq}\rho_{kbq}^{(\delta)}
\bigg).
\end{aligned}
\end{equation}
Moreover, the bound in~\eqref{approx01} is tight at $\boldsymbol{\rho}=\boldsymbol{\rho}^{(\delta)}$.
\end{lemma}

\def\QEDmark{\ensuremath{\blacksquare}}
\def\proof{\emph{Proof: }}
\def\endproof{\hfill\QEDmark}

\proof
See Appendix~\ref{appen1}.
\endproof

Using Lemma~\ref{lemma1}, we form the SCA subproblem at iteration $\delta$ by replacing $-C_{kbp}$ in~\eqref{po2} with its surrogate $-C_{kbp}^{(\delta+1)}$ in~\eqref{approx01} while keeping the original constraints~\eqref{pc21}-\eqref{pc23}:
\begin{subequations}
\label{problem3}
\begin{align}
\text{(P3)}^{(\delta+1)}:\quad
\begin{split}
\label{po3}
\underset{\boldsymbol{\rho}}{\mindot}\quad
& \sum_{k=1}^{K} \sum_{b=1}^{B} \sum_{p=1}^{P_{kb}}
- C_{kbp}^{(\delta+1)}(\boldsymbol{\rho})
\end{split}
\\
\begin{split}
\label{pc31}
\textrm{s.t.}\quad
& \sum_{b=1}^{B} \sum_{p=1}^{P_{kb}} \rho_{kbp} \leq \rho_\text{tot},\quad \forall k
\end{split}
\\
\begin{split}
\label{pc32}
& \rho_{kbp} \geq 0,\quad \forall (k,b,p)
\end{split}
\\
\begin{split}
\label{pc33}
& \gamma_{kbp}(\boldsymbol{\rho}) \geq \gamma_\text{min},\quad \forall (k,b,p)
\end{split}
\end{align}
\end{subequations}
where $\gamma_{kbp}(\boldsymbol{\rho})$ is given by~\eqref{sinr0}. By construction, the objective~\eqref{po3} and the constraints~\eqref{pc31}-\eqref{pc33} are convex in $\boldsymbol{\rho}$. Hence, each subproblem $\text{(P3)}^{(\delta+1)}$ can be solved using standard convex optimization methods~\cite{convopt}. Let $\boldsymbol{\rho}^{(\delta+1)}$ denote the resulting optimizer. We then update the surrogate parameters using~\eqref{thetat}-\eqref{zetat} with $\boldsymbol{\rho}^{(\delta+1)}$ and repeat until convergence. The SCA convergence result applies to the fixed-schedule power-allocation problem when the per-GS power and minimum-SINR constraints are feasible and the algorithm is initialized with a feasible power vector. Under these conditions, the surrogate in Lemma~\ref{lemma1} is tight and first-order consistent at the current iterate, and the standard SCA conditions in~\cite{2016oe} ensure convergence to a stationary point of~(P2). When a sampled schedule is infeasible, the SCA step is bypassed and uniform power allocation is used only as a fallback to continue the environment rollout; the event is recorded separately and is not regarded as satisfying the minimum-SINR constraints. In practical deployment, an infeasible schedule can be handled by reducing the target SINR or by rescheduling one or more UAMs to the satellite tier. This convergence statement applies only to the feasible fixed-schedule power-allocation problem and does not imply convergence of the overall DRL-SCA framework.

\subsection{General PPO Framework}
\label{genppo}

In this subsection, we briefly review the generic PPO learning framework. PPO alternates between collecting on-policy rollouts by interacting with the environment under the current policy and updating the policy and value networks for multiple epochs using the collected samples~\cite{ppopaper}. A scheduling episode consists of discrete time steps $t=0,\cdots,T_\text{s}-1$. At each step, the agent observes a state $\mathcal{S}_t$, samples an action $\mathcal{A}_t$ from the current policy $\pi_\vartheta(\mathcal{A}_t|\mathcal{S}_t)$, receives a reward $\mathcal{R}_t$, and transitions to the next state $\mathcal{S}_{t+1}$. From the rollout, we store tuples $(\mathcal{S}_t,\mathcal{A}_t,\mathcal{R}_t,d_t)$, where $d_t\in\{0,1\}$ indicates termination. In our scheduling setting, $\pi_\vartheta(\cdot|\mathcal{S}_t)$ is a categorical distribution over the discrete scheduling actions.

PPO typically uses generalized advantage estimation (GAE) to obtain low-variance advantage estimates while controlling bias using a tunable parameter~\cite{ppopaper}. Defining the temporal-difference error
\begin{equation}
\label{temporaldiff}
\delta_t = \mathcal{R}_t + \gamma_\text{D} (1-d_t)V_\varphi(\mathcal{S}_{t+1}) - V_\varphi(\mathcal{S}_t),
\end{equation}
where $0<\gamma_\text{D}<1$ is the discount factor, the GAE advantage estimate is computed recursively as
\begin{equation}
\label{advantageestimate}
\hat{A}_t = \delta_t+\gamma_\text{D}\lambda_\text{G}(1-d_t)\hat{A}_{t+1},
\end{equation}
with $\lambda_\text{G} \in [0,1]$ controlling the bias-variance tradeoff. The corresponding return target for critic regression is $\hat{\mathcal{R}}_t = \hat{A}_t + V_\varphi(\mathcal{S}_t)$. For policy updates, PPO forms the likelihood ratio
\begin{equation}
\label{likeratio}
r_t(\vartheta) = \frac{\pi_\vartheta(\mathcal{A}_t | \mathcal{S}_t)}{\pi_{\vartheta_{\text{old}}}(\mathcal{A}_t | \mathcal{S}_t)},
\end{equation}
and maximizes the clipped surrogate objective
\begin{equation}
\label{clippedobjective}
\mathcal{L}_{\mathrm{clip}}(\vartheta) =
\mathbb{E}\Big[\min\!\big(r_t(\vartheta)\hat{A}_t,\,
\mathrm{clip}(r_t(\vartheta),1-\varepsilon,1+\varepsilon)\hat{A}_t\big)\Big],
\end{equation}
where $\varepsilon \ll 1$ is the clipping parameter. The critic is trained by minimizing the mean-squared error
\begin{equation}
\label{targetmse}
\mathcal{L}_{\mathrm{v}}(\varphi) = \mathbb{E}\big[(V_\varphi(\mathcal{S}_t) - \hat{\mathcal{R}}_t)^2\big],
\end{equation}
and an entropy bonus encourages exploration:
\begin{equation}
\label{entropybonus}
\mathcal{L}_{\mathrm{ent}}(\vartheta) = \mathbb{E}\big[\mathcal{H}(\pi_\vartheta(\,\cdot\,| \mathcal{S}_t))\big].
\end{equation}
The overall PPO objective is
\begin{equation}
\label{combinedobjective}
\mathcal{J}(\vartheta,\varphi) =
\mathcal{L}_{\mathrm{clip}}(\vartheta) - c_v\,\mathcal{L}_{\mathrm{v}}(\varphi)
+ c_e\,\mathcal{L}_{\mathrm{ent}}(\vartheta),
\end{equation}
with coefficients $c_v,c_e>0$. In practice, the samples from the current rollout are partitioned into mini-batches and optimized for several epochs with gradient clipping. A new rollout is then collected under the updated policy, and this rollout-and-update cycle constitutes the PPO backbone used in our DRL-based algorithm.

\subsection{Environment and Agent Modeling for GeoSetPPO-Based UAM Scheduling}

We model the UAM scheduling problem as a Markov decision process (MDP) with state $\mathcal{S}_t$, action $\mathcal{A}_t$, transition dynamics $P(\mathcal{S}_{t+1}|\mathcal{S}_t,\mathcal{A}_t)$, and reward $\mathcal{R}_t$. On top of this environment, we design the GeoSetPPO actor $\pi_\vartheta(\mathcal{A}_t|\mathcal{S}_t)$ and critic $V_\varphi(\mathcal{S}_t)$. Note that the power-allocation module can be viewed as part of the environment: given a schedule $\mathcal{A}_t$, it maps the geometric state to per-link SINRs and, consequently, to the reward $\mathcal{R}_t$.

\begin{figure*}[t]
	\begin{center}
		{\includegraphics[width=1.9\columnwidth,keepaspectratio]
			{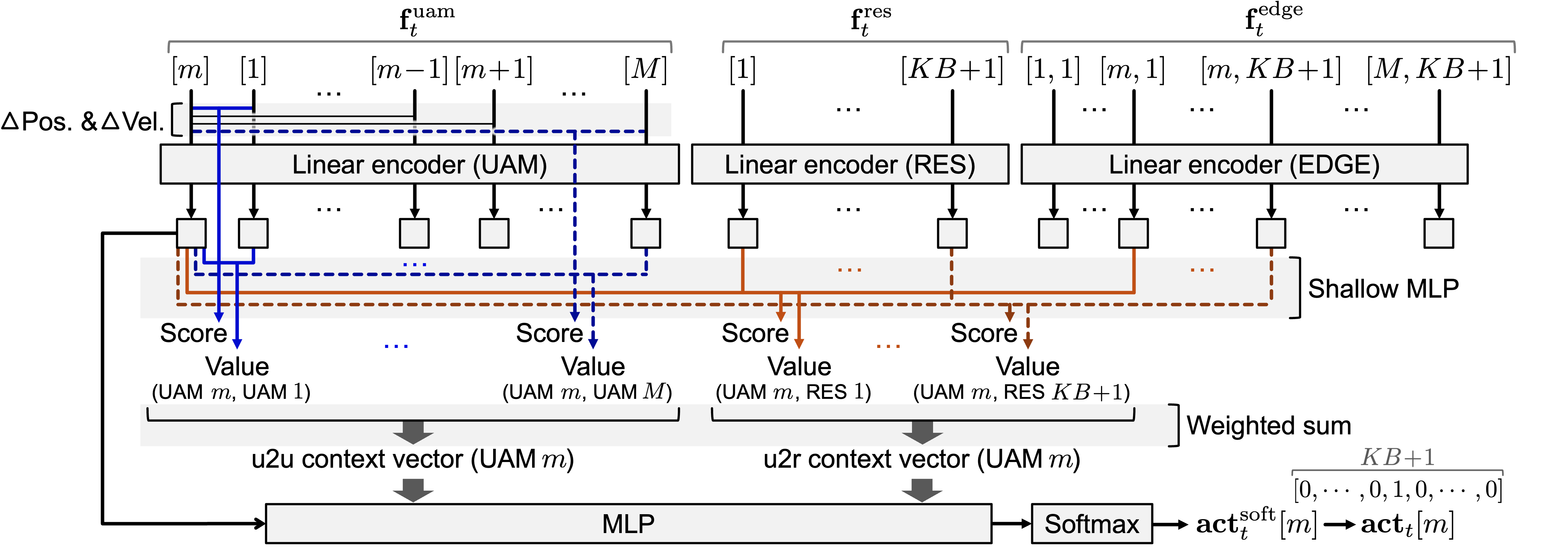}%
			\caption{Forward pass of the GeoSetPPO actor for a target UAM $m$. Set attention forms UAM-to-UAM and UAM-to-resource context vectors, then an MLP head outputs a categorical distribution over the $KB+1$ resources.}
			\label{f03}
		}
	\end{center}
	\vspace{-13pt}
\end{figure*}

\subsubsection{State Space}

The state $\mathcal{S}_t$ collects the information available to the UAM network controller at time step $t$. By assumption, the controller observes the UAM positions $\bold{u}_m[t]$ and velocities $\dot{\bold{u}}_m[t]$ for $m=1,\cdots,M$, as well as the GS positions $\bold{u}_k^\text{G}$ for $k=1,\cdots,K$. To account for handover-related effects, we additionally include the previous scheduling decision $\mathcal{A}_{t-1}$ for $t\geq 1$. For notational simplicity, we define $\mathcal{A}_{-1}$ as an arbitrary action and omit it from the network inputs at $t=0$. The state is thus given by
\begin{equation}
\label{statetuple}
\begin{aligned}
\mathcal{S}_t = \big(
&\bold{u}_1[t],\cdots,\bold{u}_M[t],
\dot{\bold{u}}_1[t],\cdots,\dot{\bold{u}}_M[t],\\
&\bold{u}_1^\text{G},\cdots,\bold{u}_K^\text{G},
\mathcal{A}_{t-1}
\big).
\end{aligned}
\end{equation}
The agent observes $\mathcal{S}_t$ and transforms it into feature representations that serve as inputs to the GeoSetPPO actor and critic models.

\subsubsection{Action Space}

As described in Section~\ref{schedprob}, we treat each GS-subband pair and the satellite as a resource. An action $\mathcal{A}_t$ specifies, for all UAMs, which resource to connect to. Since each of the $K$ GSs offers $B$ subbands and the satellite has no subband choice, there are $KB+1$ resources in total. For UAM $m$ at time step $t$, the scheduling decision is represented by a one-hot vector $\mathbf{act}_t[m]\in\{0,1\}^{KB+1}$. If UAM $m$ is associated with GS $k$ on subband $b$, then the $((k-1)B+b)$-th entry of $\mathbf{act}_t[m]$ is $1$. If it is associated with the satellite, then the $(KB+1)$-th entry is $1$. The resulting joint action at time step $t$ is
\begin{equation}
\label{actiontuple}
\mathcal{A}_t=\big(\mathbf{act}_t[1],\cdots,\mathbf{act}_t[M]\big),
\end{equation}
which encodes the scheduling decisions for all UAMs at time step $t$.

\subsubsection{Reward}

We adopt the per-step objective of (P1) as the reward. From~\eqref{po1}, the reward at time step $t$ is
\begin{equation}
\label{rewarddef}
\mathcal{R}_t
=
\frac{1}{B}
\sum_{k=1}^{K}
\sum_{b=1}^{B}
\sum_{p=1}^{P_{kb}}
C_{kbp}[t]
+\mathcal{R}_t^\text{sat}
-\mathcal{P}_t^\text{ho}
-\mathcal{P}_t^\text{gs},
\end{equation}
where $\mathcal{R}_t^\text{sat}$, $\mathcal{P}_t^\text{ho}$, and $\mathcal{P}_t^\text{gs}$ are computed using~\eqref{pc11}, \eqref{pc12}, and \eqref{pc13}, respectively. Hence, $\mathcal{R}_t$ combines the GS-tier throughput and satellite-tier throughput with penalties for handovers and GS overload. The relative strengths of these penalties are controlled by $c^\text{B}$, $c^\text{G}$, $c^\text{S}$, and $c_\text{gs}$. Since handover cost typically increases from inter-subband changes to inter-GS handovers and then to GS-satellite handovers, we typically set $0\leq c^\text{B}<c^\text{G}<c^\text{S}$.

\subsubsection{State Features}

We convert $\mathcal{S}_t$ into three feature sets: UAM features, resource features, and edge (UAM-resource) features:
\begin{equation}
\label{threefeats}
\begin{aligned}
\mathbf{F}_t^\text{uam}&=(\mathbf{f}_t^\text{uam}[1], \cdots,\mathbf{f}_t^\text{uam}[M]),\\
\mathbf{F}_t^\text{res}&=(\mathbf{f}_t^\text{res}[1], \cdots,\mathbf{f}_t^\text{res}[KB+1]),\\
\mathbf{F}_t^\text{edge}&=(\mathbf{f}_t^\text{edge}[1,1], \cdots,\mathbf{f}_t^\text{edge}[M,KB+1]).
\end{aligned}
\end{equation}
Here, $KB+1$ is the total number of resources. For the $m$-th UAM, we define
\begin{equation}
\label{uamfeats}
\mathbf{f}_t^\text{uam}[m]=\big[\bold{u}_m[t];\dot{\bold{u}}_m[t];\mathbf{act}_{t-1}[m]\big]
\in \mathbb{R}^{KB+7},
\end{equation}
where $\mathbf{act}_{t-1}[m]$ is the previous scheduling decision for that UAM. Let $i=(k-1)B+b$ with $k\in\{1,\cdots,K\}$ and $b\in\{1,\cdots,B\}$ so that the $i$-th resource corresponds to the $k$-th GS and $b$-th subband.
The corresponding resource feature is
\begin{equation}
\label{gsfeats}
\mathbf{f}_t^\text{res}[i]=\big[\mathbf{idx}_t^{\text{GS},k};\mathbf{idx}_t^{\text{band},b};\bold{u}_k^\text{G};0\big]\in\mathbb{R}^{K+B+4},
\end{equation}
where $\mathbf{idx}_t^{\text{GS},k}\in\mathbb{R}^{K}$ and $\mathbf{idx}_t^{\text{band},b}\in\mathbb{R}^{B}$ are one-hot encodings of the GS and subband indices, respectively. For the satellite resource ($i=KB+1$), we simply use a one-hot ``satellite flag":
\begin{equation}
\label{satfeats}
\mathbf{f}_t^\text{res}[KB+1]=[0,\cdots,0,1]^\text{T}\in\mathbb{R}^{K+B+4}.
\end{equation}
Using the same mapping $i=(k-1)B+b$ with $k\in\{1,\cdots,K\}$ and $b\in\{1,\cdots,B\}$, the edge feature between the $m$-th UAM and the $i$-th resource is
\begin{equation}
\label{edgegsfeats}
\mathbf{f}_t^\text{edge}[m,i]=
\big[\|\bold{u}_m[t]-\bold{u}_k^\text{G}\|^2,\;\mathrm{IF}_t^{m,k}\big]^\text{T}\in\mathbb{R}^{2},
\end{equation}
which captures both link quality (squared distance) and a heuristic interference score. Specifically, $\mathrm{IF}_t^{m,k}$ quantifies the potential interference that would be induced if GS $k$ serves UAM $m$ at time step $t$:
\begin{equation}
\label{ifcause}
\begin{aligned}
\mathrm{IF}_t^{m,k}
= \sum_{n\neq m}^{M}
&\Big|\bold{a}_{N_x}(\cos\alpha_{m}^{(k)}[t])^\text{H}\bold{a}_{N_x}(\cos\alpha_{n}^{(k)}[t])\\
&\cdot
\bold{a}_{N_y}(\cos\beta_{m}^{(k)}[t])^\text{H}\bold{a}_{N_y}(\cos\beta_{n}^{(k)}[t])\Big|^2 \\
&\times \frac{1}{\|\bold{u}_n[t]-\bold{u}_k^\text{G}\|^2}.
\end{aligned}
\end{equation}
For the satellite resource ($i=KB+1$), we set
\begin{equation}
\label{edgesatfeats}
\mathbf{f}_t^\text{edge}[m,KB+1]=[0,0]^\text{T},
\end{equation}
since the satellite uses an orthogonal band and the satellite-tier sum throughput is fixed to $C_\text{SAT}$. $\\$

\begin{algorithm}[t]
\caption{Two-Stage GeoSetPPO Training}
\label{alg01}
\begin{algorithmic}[1]
\Require Environment model, GeoSetPPO actor $\pi_\vartheta$, critic $V_\varphi$, PPO hyperparameters $(\gamma_\text{D},\lambda_\text{G},\varepsilon,c_v,c_e)$, SCA power-allocation module, warm-up iterations $N_{\text{uni}}$, transition horizon $N_{\text{mix}}$, total iterations $N_{\text{train}}$
\State Initialize actor parameters $\vartheta$ and critic parameters $\varphi$
\For{PPO iteration $r = 1,\cdots,N_{\text{train}}$}
    \State Set reference policy $\vartheta_{\text{old}} \gets \vartheta$
    \If{$r \le N_{\text{uni}}$}
        \State $\eta_r \gets 1$
    \Else
        \State $\eta_r \gets \max\{0,\, 1 - (r - N_{\text{uni}})/N_{\text{mix}}\}$
    \EndIf
    \State $\mathcal{D} \leftarrow \emptyset$
    \For{episode $e_\text{epi} = 1,\cdots,N_{\text{epi}}$}
        \State Reset environment and obtain initial state $\mathcal{S}_0$
        \For{$t = 0,\cdots,T_s-1$}
            \State Construct state $\mathcal{S}_t$ from~\eqref{statetuple}
            \State Sample scheduling action $\mathcal{A}_t \sim \pi_{\vartheta_{\text{old}}}(\cdot | \mathcal{S}_t)$
            \State Compute uniform powers $\boldsymbol{\rho}_t^{\text{uni}}$
            \If{$r \le N_{\text{uni}}$}
                \State $\boldsymbol{\rho}_t \gets \boldsymbol{\rho}_t^{\text{uni}}$
            \Else
                \State $\boldsymbol{\rho}_t^{\text{sca}} \gets \text{SCA-Power-Allocation}(\mathcal{S}_t,\mathcal{A}_t)$
                \State $\boldsymbol{\rho}_t \gets \eta_r \boldsymbol{\rho}_t^{\text{uni}} + (1-\eta_r)\boldsymbol{\rho}_t^{\text{sca}}$
            \EndIf
            \State Evaluate reward $\mathcal{R}_t$ using $\boldsymbol{\rho}_t$
            \State Step environment to $\mathcal{S}_{t+1}$ with terminal flag $d_t$
            \State Store $(\mathcal{S}_t,\mathcal{A}_t,\mathcal{R}_t,\mathcal{S}_{t+1},d_t,
                   \log\pi_{\vartheta_{\text{old}}}(\mathcal{A}_t|\mathcal{S}_t),$ $V_\varphi(\mathcal{S}_t))$ in $\mathcal{D}$
        \EndFor
    \EndFor
    \State Compute $\delta_t$ by \eqref{temporaldiff} and $\hat{A}_t$ by \eqref{advantageestimate} on $\mathcal{D}$
    \State Set $\hat{\mathcal{R}}_t = \hat{A}_t + V_\varphi(\mathcal{S}_t)$ and normalize $\{\hat{A}_t\}$
    \For{epoch $e\!=\!1,\cdots,N_{\text{epo}}$ and mini-batch $\mathcal{B}\!\subset\!\mathcal{D}$}
            \State Recompute $\log \pi_\vartheta(\mathcal{A}_t|\mathcal{S}_t)$ and $V_\varphi(\mathcal{S}_t)$ for $t\in\mathcal{B}$
            \State $r_t(\vartheta) \gets \exp\big(\log\pi_\vartheta(\mathcal{A}_t|\mathcal{S}_t) - \log\pi_{\vartheta_{\text{old}}}(\mathcal{A}_t|\mathcal{S}_t)\big)$
            \State Set $\mathcal{L}_{\mathrm{clip}}(\vartheta)$, $\mathcal{L}_{\mathrm{v}}(\varphi)$, $\mathcal{L}_{\mathrm{ent}}(\vartheta)$ as in \eqref{clippedobjective}, \eqref{targetmse}, \eqref{entropybonus}
            \State Combine into total objective $\mathcal{J}(\vartheta,\varphi)$ by \eqref{combinedobjective}
            \State Update $\vartheta$ and $\varphi$ by gradient-ascent on $\mathcal{J}(\vartheta,\varphi)$
    \EndFor
    \State \textbf{if} performance converged \textbf{then break}
\EndFor
\State \Return trained GeoSetPPO actor $\pi_\vartheta$ and critic $V_\varphi$
\end{algorithmic}
\end{algorithm}
\vspace{-7pt}

\subsubsection{Actor}

The feature sets $\mathbf{F}_t^\text{uam}$, $\mathbf{F}_t^\text{res}$, and $\mathbf{F}_t^\text{edge}$ are fed to the GeoSetPPO actor. 
Rather than explicitly parameterizing a categorical distribution over all $(KB+1)^M$ joint assignments, the actor uses a shared per-UAM categorical head and produces $M(KB+1)$ logits. The factorization is applied only at the output, since each per-UAM distribution is conditioned on the complete UAM, resource, and edge sets, and the assembled joint schedule is evaluated using a common reward that reflects the satellite-tier contribution, subband interference, GS overload, and SCA-based power allocation. Thus, the factorized policy provides a scalable approximation to the coupled scheduling problem, although it does not impose hard feasibility constraints on the sampled joint action.
Consequently, scheduling decisions are driven by relative geometry and UAM-resource interactions and are not affected by the arbitrary ordering of UAM indices. To this end, we employ a set-attention architecture that incorporates geometry-derived pairwise information in attention scoring, such as relative position and relative velocity~\cite{setattention}. In particular, for a fixed target UAM~$m$, the actor output is invariant to any permutation of the remaining UAM or resource tokens.\footnote{This permutation invariance concerns arbitrary token ordering and does not imply zero-shot transfer across different network dimensions, GS layouts, mobility distributions, subband configurations, or reward definitions, which may require retraining or fine-tuning.} The overall actor architecture is illustrated in Fig.~\ref{f03}. For each time step $t$ and $\forall m,i$, we first pass the raw features through linear encoders
\begin{equation}
\label{linearfeats}
\begin{aligned}
\mathbf{g}_t^\text{uam}[m]&=\mathrm{enc}_\text{u}(\mathbf{f}_t^\text{uam}[m]),\\
\mathbf{g}_t^\text{res}[i]&=\mathrm{enc}_\text{r}(\mathbf{f}_t^\text{res}[i]),\\
\mathbf{g}_t^\text{edge}[m,i]&=\mathrm{enc}_\text{e}(\mathbf{f}_t^\text{edge}[m,i]),
\end{aligned}
\end{equation}
where $\mathbf{g}_t^\text{uam}[m],\mathbf{g}_t^\text{res}[i],\mathbf{g}_t^\text{edge}[m,i]\in\mathbb{R}^{H}$ and $H$ is the hidden dimension. We then apply Bahdanau-style attention~\cite{bahdanaupaper} to capture UAM-UAM and UAM-resource interactions. The attention scores $\mathrm{s}_t^\text{u2u}[m,n]$ and $\mathrm{s}_t^\text{u2r}[m,i]$ are computed using shallow MLPs, $\mathrm{score}_\text{u2u}$ and $\mathrm{score}_\text{u2r}$:
\begin{equation}
\label{attscore}
\begin{aligned}
\mathrm{s}_t^\text{u2u}[m,n]
&=\mathrm{score}_\text{u2u}\big([\mathbf{g}_t^\text{uam}[m];\mathbf{g}_t^\text{uam}[n];
\Delta\bold{u}_{m,n}[t];\\
&\hspace{1.4cm}\Delta\dot{\bold{u}}_{m,n}[t];
\bold{u}_{m}[t];\dot{\bold{u}}_{m}[t]]\big),\\
\mathrm{s}_t^\text{u2r}[m,i]
&=\mathrm{score}_\text{u2r}\big([\mathbf{g}_t^\text{uam}[m];
\mathbf{g}_t^\text{res}[i];\mathbf{g}_t^\text{edge}[m,i]]\big),
\end{aligned}
\end{equation}
where $\Delta\bold{u}_{m,n}[t]=\bold{u}_{n}[t]-\bold{u}_{m}[t]$ and $\Delta\dot{\bold{u}}_{m,n}[t]=\dot{\bold{u}}_{n}[t]-\dot{\bold{u}}_{m}[t]$ are the relative position and velocity, respectively. Attention weights are obtained via softmax:
\begin{equation}
\label{attweight}
\begin{aligned}
\mathrm{w}_t^\text{u2u}[m,n]&=\underset{n\in \{1,\cdots,M\}}{\mathrm{softmax}}(\mathrm{s}_t^\text{u2u}[m,n]),\\
\mathrm{w}_t^\text{u2r}[m,i]&=\underset{i\in \{1,\cdots,KB+1\}}{\mathrm{softmax}}(\mathrm{s}_t^\text{u2r}[m,i]).
\end{aligned}
\end{equation}
The corresponding value vectors are
\begin{equation}
\label{attvalue}
\begin{aligned}
\mathbf{v}_t^\text{u2u}[m,n]
&=\mathrm{value}_\text{u2u}\big([\mathbf{g}_t^\text{uam}[n];
\Delta\bold{u}_{m,n}[t];
\Delta\dot{\bold{u}}_{m,n}[t]]\big),\\
\mathbf{v}_t^\text{u2r}[m,i]
&=\mathrm{value}_\text{u2r}\big([\mathbf{g}_t^\text{res}[i];
\mathbf{g}_t^\text{edge}[m,i]]\big),
\end{aligned}
\end{equation}
where $\mathrm{value}_\text{u2u}$ and $\mathrm{value}_\text{u2r}$ are also shallow MLPs. The UAM-UAM and UAM-resource context vectors for UAM $m$ are then
\begin{equation}
\label{contexts}
\begin{aligned}
\mathbf{ctx}_t^\text{u2u}[m]&=\sum_{n=1}^{M}\mathrm{w}_t^\text{u2u}[m,n]\mathbf{v}_t^\text{u2u}[m,n],\\
\mathbf{ctx}_t^\text{u2r}[m]&=\sum_{i=1}^{KB+1}\mathrm{w}_t^\text{u2r}[m,i]\mathbf{v}_t^\text{u2r}[m,i].
\end{aligned}
\end{equation}
Finally, the actor produces a categorical distribution over the $KB+1$ resources:
\begin{equation}
\label{policyout}
\begin{aligned}
\mathbf{act}_t^\text{soft}[m]
=\mathrm{softmax}\big(
\mathrm{lgt}([\mathbf{g}_t^\text{uam}[m];
\mathbf{ctx}_t^\text{u2u}[m];
\mathbf{ctx}_t^\text{u2r}[m]])
\big),
\end{aligned}
\end{equation}
where $\mathrm{lgt}$ is an MLP that outputs $KB+1$ logits, and $\mathbf{act}_t^\text{soft}[m]\in\mathbb{R}^{KB+1}$. During training, the one-hot action vector $\mathbf{act}_t[m]$ is sampled from the categorical distribution defined by $\mathbf{act}_t^\text{soft}[m]$, whereas during evaluation we select the resource with the largest probability. Aggregating the per-UAM decisions into $\mathcal{A}_t$ as in~\eqref{actiontuple} yields the joint policy $\pi_\vartheta(\mathcal{A}_t|\mathcal{S}_t)$.

\begin{figure}[t]
	\begin{center}
		{\includegraphics[width=0.95\columnwidth,keepaspectratio]
			{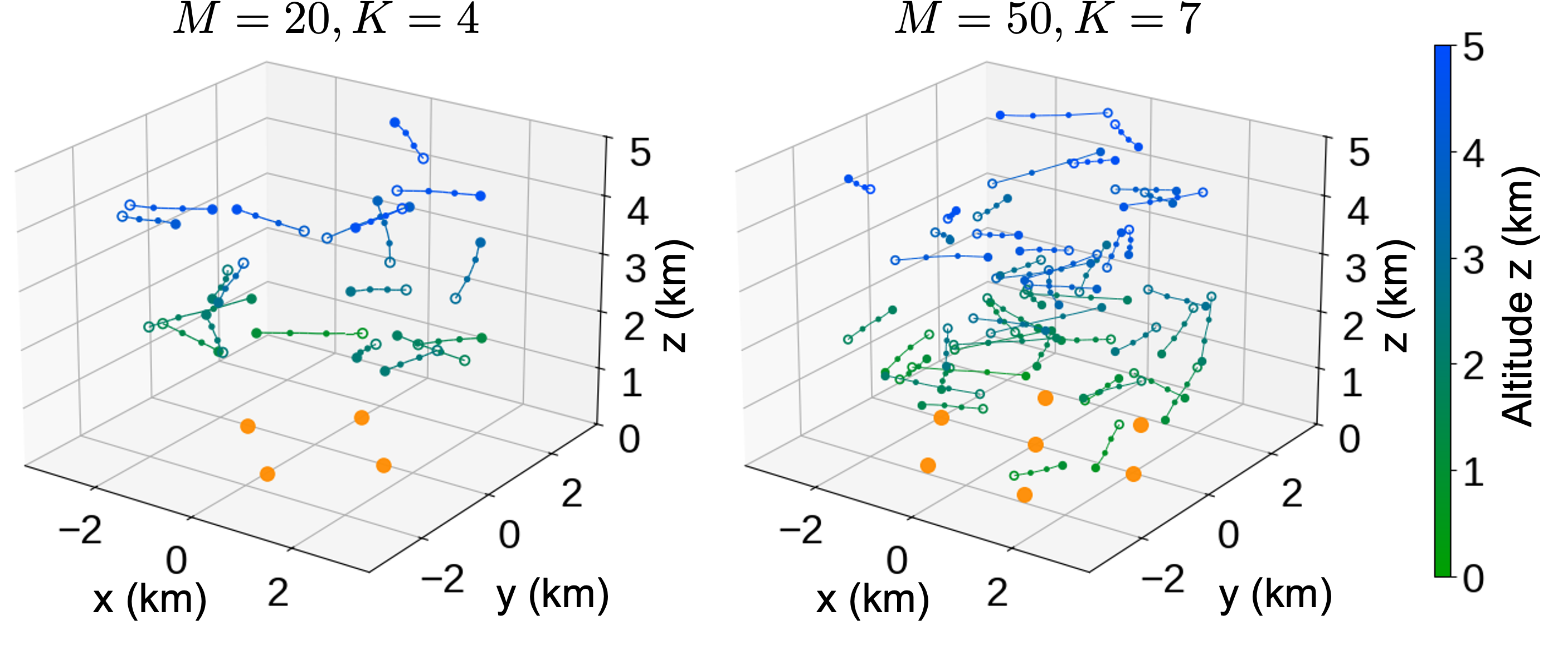}%
			\caption{Example 3D deployments and UAM trajectories. Filled markers and hollow markers denote UAM locations at $t=0$ and $t=T_\text{s}-1$. Smaller markers indicate intermediate trajectory samples every $20$~s, and orange markers denote GS locations.}
			\label{fig06}
		}
	\end{center}
	\vspace{-10pt}
\end{figure}

\subsubsection{Critic}

The critic network $V_\varphi(\mathcal{S}_t)$ estimates the value of the current state by aggregating information from all UAMs, resources, and edges. Each feature type is first encoded and then processed through attention and pooling modules to obtain permutation-invariant summaries. These summaries are concatenated and fed to an MLP value head to produce the scalar estimate $V_\varphi(\mathcal{S}_t)$. In this way, the critic shares the same inductive bias as the actor, capturing geometry-driven interactions among UAMs and available resources.

\subsection{Combined Training Algorithm for GeoSetPPO}
\label{sec_combalg}

Building on the proposed GeoSetPPO actor-critic architectures, we train the scheduling agent using the PPO framework summarized in Section~\ref{genppo}. The GeoSetPPO actor implements the policy $\pi_\vartheta(\mathcal{A}_t|\mathcal{S}_t)$, and the critic $V_\varphi(\mathcal{S}_t)$ provides baseline estimates for advantage computation.
Training alternates between collecting on-policy rollouts and updating the actor and critic: temporal-difference errors $\delta_t$ are computed via~\eqref{temporaldiff}, advantages $\hat{A}_t$ are estimated using GAE in~\eqref{advantageestimate}, and network parameters are updated by maximizing the PPO objective in~\eqref{combinedobjective}.

On top of this standard PPO backbone, we introduce a two-stage training strategy that controls the power-allocation rule used in reward evaluation, as illustrated in Fig.~\ref{f02}. In the first stage, we use uniform power allocation after each scheduling decision. Given $\mathcal{A}_t$, transmit power is obtained by equally splitting the per-GS power budget across the scheduled links, and the reward $\mathcal{R}_t$ is computed accordingly. This stage runs for $N_{\text{uni}}$ PPO iterations, avoiding nonlinear-programming overhead and exposing both the actor and critic to a simple reward structure.

{\renewcommand{\arraystretch}{1.02}
\begin{table}[!t]
	\centering
	\caption{Simulation parameters}
	\scriptsize
	\label{tbl01}
	\setlength{\tabcolsep}{3pt}
	\begin{tabular}{p{0.61\columnwidth}|p{0.33\columnwidth}}
	\hline
	\bfseries{Parameter} & \bfseries{Value} \\
	\hhline{=|=}
	Number of UAMs and GSs ($M,K$) & $(20,4)$, $(50,7)$\\
	\hline
	Per-GS user threshold ($\mathcal{N}_\text{GS}$) & $4$, $6$\\
	\hline
	Number of GS subbands ($B$) & $1$, $2$\\
	\hline
	Size of GS antenna array ($N_x,N_y$) & $(4,4)$, $(6,6)$, $(10,10)$\\
	\hline
	GS transmit power budget ($\rho_\text{tot}$) & $1$~\si{\watt}\\
	\hline
	GS carrier frequency ($c/\lambda$) & $7.5$~\si{\giga\hertz}\\
	\hline
	Total GS-tier bandwidth ($W_\text{G}$) & $100$~\si{\mega\hertz}\\
    \hline
    Minimum SINR threshold ($\gamma_\text{min}$)	& $0$\,\,\si{dB}\\
	\hline
    GS-UAM Rician factor ($K_\text{R}$) & $20$~\si{dB}\\
    \hline
    Satellite carrier frequency and bandwidth ($c/\lambda_\text{S},W_\text{S}$) & $20$~\si{\giga\hertz}, $100$~\si{\mega\hertz}\\
    \hline
    LEO altitude and spot-beam transmit power ($z^\text{S},\rho_\text{S}$) & $600$~\si{\kilo\metre}, $50$~\si{\watt}\\
    \hline
    Satellite Tx and UAM Rx gains ($G_\text{S}^\text{T}(0),G_\text{S}^\text{R}$) & $37.39$, $20$~\si{dBi}\\
    \hline
    Other satellite-link loss ($\mathcal{L}_\text{S}$) & $10$~\si{dB}\\
	\hline
	Satellite spectral efficiency ($C_\text{SAT}$) & $2$~\si{bps/Hz}\\
	\hline
	Time horizon ($\tau_\text{s}T_\text{s}$) & $60$~\si{\second}\\
	\hline
	Time slot duration ($\tau_\text{s}$) & $5$~\si{\second}\\
	\hline
	Max. and min. altitude of UAM & $5$, $0.5$~\si{\kilo\metre}\\
	\hline
	Max. and min. speed of UAM & $50$, $10$~\si{\metre/\second}\\
	\hline
	Minimum GS separation & $2$~\si{\kilo\metre}\\
	\hline
	Penalty coefficients ($c^\text{B},c^\text{G},c^\text{S},c_\text{gs}$) & $0.2$, $0.6$, $1.0$, $0.5$\\
	\hline
    PPO learning rate & $10^{-4}$\\
    \hline
    PPO mini-batch size & $256$\\
    \hline
    PPO update epochs ($N_\text{epo}$) & $4$\\
    \hline

    PPO discount factor ($\gamma_\text{D}$)
    & $0.99$\\

    \hline

    GAE parameter ($\lambda_\text{G}$)
    & $0.95$\\

    \hline

    PPO clipping parameter ($\varepsilon$)
    & $0.1$\\

    \hline

	Warm-up steps ($N_\text{uni}N_\text{epi}T_\text{s}$) & $10$~\si{M}, $14$~\si{M}\\
	\hline
	Transition steps ($N_\text{mix}N_\text{epi}T_\text{s}$) & $50000$\\
	\hline
	Total steps ($N_\text{train}N_\text{epi}T_\text{s}$) & $10.4$~\si{M}, $14.4$~\si{M}\\
	\hline
	\end{tabular}
	\vspace{-1pt}
\end{table}
}

After the first stage, we gradually introduce the SCA-based power allocator into reward computation. Let $r$ denote the PPO iteration index. For $r > N_{\text{uni}}$, we define a mixing weight $\eta_r \in [0,1]$ that decreases linearly over a transition horizon of $N_{\text{mix}}$ iterations. At each iteration $r$, we compute the uniform power vector $\boldsymbol{\rho}_t^{\text{uni}}$ and also obtain the SCA solution $\boldsymbol{\rho}_t^{\text{sca}}$ for the scheduled links. The power vector used for rate and reward evaluation is then the linear combination
\begin{equation}
\label{rho_mix}
\boldsymbol{\rho}_t = \eta_r \boldsymbol{\rho}_t^{\text{uni}} + (1-\eta_r)\boldsymbol{\rho}_t^{\text{sca}}.
\end{equation}
For $r \ge N_{\text{uni}} + N_{\text{mix}}$, we have $\eta_r=0$ and reward evaluation becomes fully SCA-based.

The overall training process is described in Algorithm~\ref{alg01}. The proposed two-stage training strategy has two main advantages. First, it avoids repeatedly invoking the SCA-based power allocator, which is computationally expensive, during the early phase of training when many PPO iterations are required. Second, the gradual transition from uniform to SCA-based power allocation prevents an abrupt change in the reward distribution, which helps the critic track the return targets and improves the stability of PPO updates. As a result, the actor can progressively exploit the performance gains of the SCA solution while maintaining stable learning dynamics.

\section{Simulation Results}
\label{sec_simres}

Table~\ref{tbl01} summarizes the main simulation parameters. The total GS-tier bandwidth is fixed at $100$~MHz for all configurations and is equally divided among the $B$ GS subbands. Thus, each subband occupies $100/B$~MHz, and its noise power is calculated over the corresponding subband bandwidth. We consider four representative settings with $(M,K,\mathcal{N}_\text{GS},B,N_x,N_y)=(20,4,4,1,6,6)$, $(20,4,4,2,4,4)$, $(50,7,6,1,10,10)$, and $(50,7,6,2,6,6)$. The multi-subband cases allow the scheduler to exploit frequency-domain separation for interference management, while the $B=1$ cases provide direct comparisons with the single-band association methods. Slightly larger GS arrays are used for $B=1$ because all GS transmissions share the same band and the resulting system is more strongly interference limited. The satellite parameters follow the link-budget setting in~\cite{myjsac}. They yield a nominal boresight SNR of approximately $14$~dB before system-level implementation margins and access overhead, for which the effective sum spectral efficiency $C_\text{SAT}=2$~bps/Hz is a conservative operating value.

\begin{figure}[t]
	\begin{center}
		{\includegraphics[width=0.95\columnwidth,keepaspectratio]
			{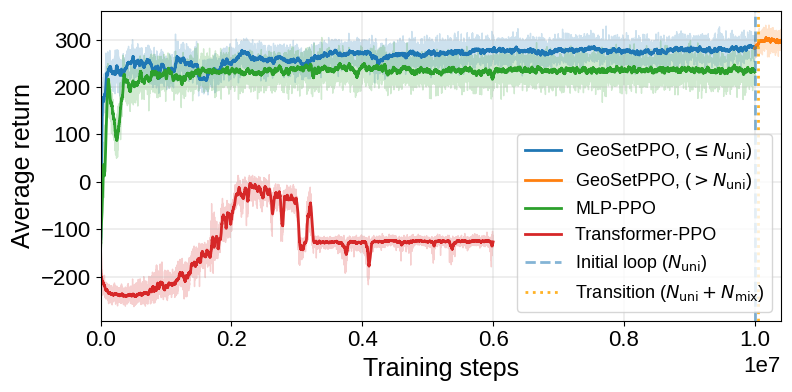}%
			\caption{Training curves of GeoSetPPO, MLP-PPO, and Transformer-PPO for $(M,K,B)=(20,4,2)$. Light curves show episode returns, bold curves show moving averages, and the vertical lines indicate the warm-up and transition boundaries of the two-stage GeoSetPPO training.}
			\label{fig:training}
		}
	\end{center}
	\vspace{-10pt}
\end{figure}

We compare GeoSetPPO with four scheduling baselines. MLP-PPO uses conventional multilayer perceptron (MLP) actor and critic networks without the proposed set-attention structure. Transformer-PPO encodes the UAM and resource features using a standard Transformer, but does not include the explicit relative-position, relative-velocity, and UAM-resource edge representations used by GeoSetPPO. The algorithm-based scheduler in~\cite{myjsac} and the distance-based heuristic determine only BS-UAM association. The distance-based method associates up to $\mathcal{N}_\text{GS}$ nearby UAMs with each GS and assigns the remaining UAMs to the satellite. The two association baselines determine BS-UAM association independently at each slot. For $B=2$, the resulting association is complemented by a deterministic persistent per-GS round-robin (RR) subband rule. UAMs that remain associated with the same GS retain their previous subbands, while newly associated UAMs are ordered by their azimuths around that GS and assigned sequentially to the least-loaded subband. This rule balances the subband loads and preserves subband assignments for retained links, but does not use channel gains, interference, or rate evaluations to select subbands.

\begin{figure}[t]
	\begin{center}
		{\includegraphics[width=1.0\columnwidth,keepaspectratio]
			{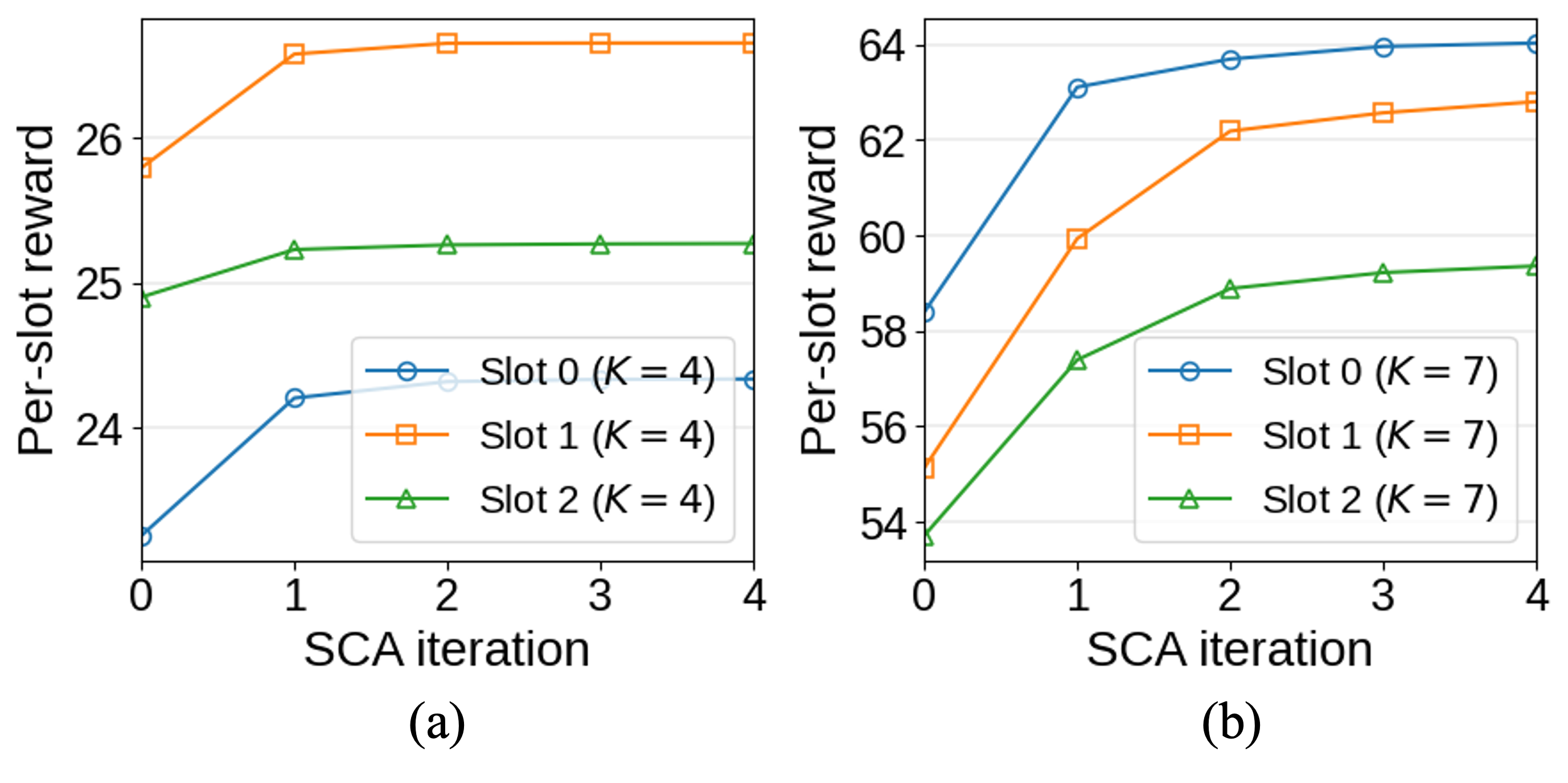}%
			\caption{Convergence of the SCA-based GS power allocation for $B=2$ under the minimum SINR constraint. Results are shown for the first three slots of randomly generated episodes with (a) $K=4$ and (b) $K=7$.}
			\label{fig:sca_convergence}
		}
	\end{center}
	\vspace{-10pt}
\end{figure}

\begin{figure}[t]
	\begin{center}
		{\includegraphics[width=1.0\columnwidth,keepaspectratio]
			{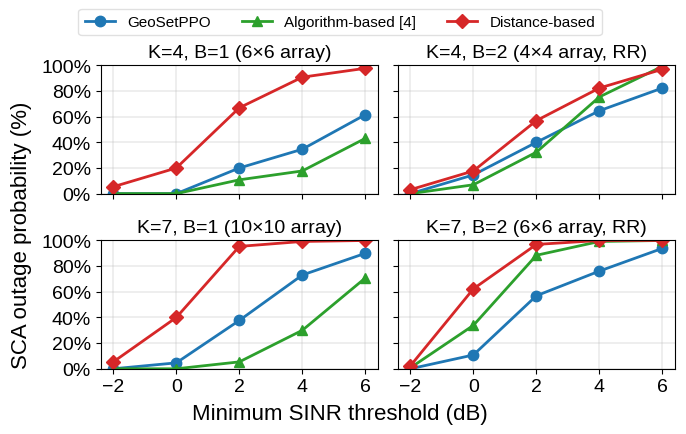}%
			\caption{SCA outage probability versus the minimum SINR threshold for GeoSetPPO, the algorithm-based scheduler, and the distance-based scheduler. RR subband allocation is applied to the two association baselines when $B=2$.}
			\label{fig:sca_outage}
		}
	\end{center}
	\vspace{-10pt}
\end{figure}

We generate random three-dimensional scenarios with UAM trajectories over a $60$~s horizon, consisting of $12$ decision steps with a $5$~s interval. At the considered speed range of $10$--$50$~m/s, a UAM travels $50$--$250$~m during one decision interval and approximately $0.6$--$3$~km over one episode. These displacements are comparable to the $2$-km minimum GS separation and therefore produce meaningful changes in link distances, beam directions, interference relationships, and handover decisions. Accordingly, the considered horizon is sufficient to evaluate the temporal coupling targeted in this work. Each UAM starts from a random initial position with a random initial velocity and a random acceleration process, which produces nontrivial trajectories within the target airspace. The GS locations follow fixed patterns, and the distance between the closest pair of GSs is $2$~km. Fig.~\ref{fig06} visualizes example deployments and trajectories for $(M,K)=(20,4)$ and $(50,7)$. Unless otherwise stated, the GS-to-UAM evaluation channels are strongly LoS dominant with a Rician $K$ factor of $20$~dB. Training and test episodes are generated independently using disjoint random seeds. The test set therefore contains unseen UAM initial states, trajectories, acceleration processes, and channel realizations drawn from the same network configuration. For each trained policy and operating point, the reported mean and its $95\%$ confidence interval are computed across independently generated test episodes with different UAM initial states, trajectories, acceleration processes, and Rician channel realizations.

The GeoSetPPO actor architecture is illustrated in Fig.~\ref{f03}. The score and value networks, $\mathrm{score}_\text{u2u}$, $\mathrm{score}_\text{u2r}$, $\mathrm{value}_\text{u2u}$, and $\mathrm{value}_\text{u2r}$, are MLPs with one hidden layer, and the policy head $\mathrm{lgt}$ is an MLP with four hidden layers. All hidden layers have width $256$, including the encoded feature dimension of $\mathbf{g}_t^\text{uam}[m]$, $\mathbf{g}_t^\text{res}[i]$, and $\mathbf{g}_t^\text{edge}[m,i]$. The MLP-PPO actor and critic each use four hidden layers of width $256$. The Transformer-PPO baseline receives the same per-UAM position, velocity, and previous-action information as token features, but relies on generic token attention without the geometry-derived pairwise scoring used by GeoSetPPO.

{\renewcommand{\arraystretch}{1.05}
\begin{table}[!t]
	\centering
	\caption{Runtime per scheduling instance and SCA power-allocation call (ms)}
	\footnotesize
	\label{tbl02}
	\setlength{\tabcolsep}{2.2pt}
	\begin{tabular}{cc|ccc||c}
	\hline
	$(M,K,B)$ & Statistic & GeoSetPPO & Alg.+RR & Dist.+RR & SCA\\
	\hline\hline
	$(20,4,2)$ & Mean & $2.782$ & $5.545$ & $1.198$ & $297.6$\\
	$(20,4,2)$ & Std. dev. & $0.136$ & $1.883$ & $1.000$ & $83.5$\\
	\hline
	$(50,7,2)$ & Mean & $2.899$ & $40.84$ & $1.340$ & $1871.8$\\
	$(50,7,2)$ & Std. dev. & $0.881$ & $2.064$ & $0.439$ & $399.9$\\
	\hline
	\end{tabular}
	\vspace{-1pt}
\end{table}
}

\begin{figure*}[t]
	\begin{center}
		{\includegraphics[width=1.8\columnwidth,keepaspectratio]
			{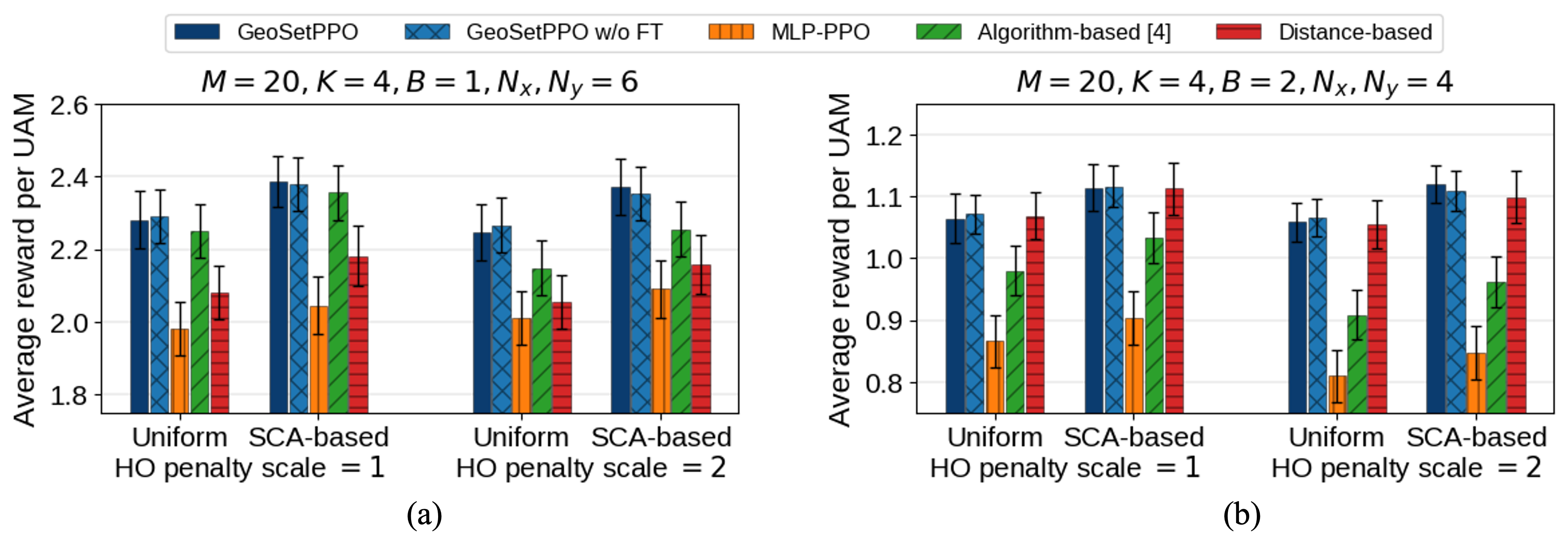}%
			\caption{Average reward per UAM for $(M,K)=(20,4)$ under uniform and SCA-based power allocation. The handover-penalty scale multiplies $(c^\text{B},c^\text{G},c^\text{S})=(0.2,0.6,1.0)$. ``GeoSetPPO w/o FT'' denotes the ablation without the final SCA-based fine-tuning stage.}
			\label{fig:reward_comparison}
		}
	\end{center}
	\vspace{-10pt}
\end{figure*}

\subsection{Convergence and Computational Analysis}

Fig.~\ref{fig:training} compares the learning behavior of the three PPO architectures. GeoSetPPO rapidly learns a stable policy during the uniform-power stage and achieves a higher return than MLP-PPO. After the transition to SCA-based reward evaluation, its return increases further without an abrupt loss of training stability. In contrast, Transformer-PPO does not obtain a competitive policy within the conducted training run. Although generic self-attention can exchange information among UAM and resource tokens, it does not explicitly represent the relative geometry that determines directional interference. Under the considered architecture and training budget, this result indicates that the explicit geometry-aware decomposition facilitates policy learning relative to generic token attention.

\begin{figure}[t]
	\begin{center}
		{\includegraphics[width=0.9\columnwidth,keepaspectratio]
			{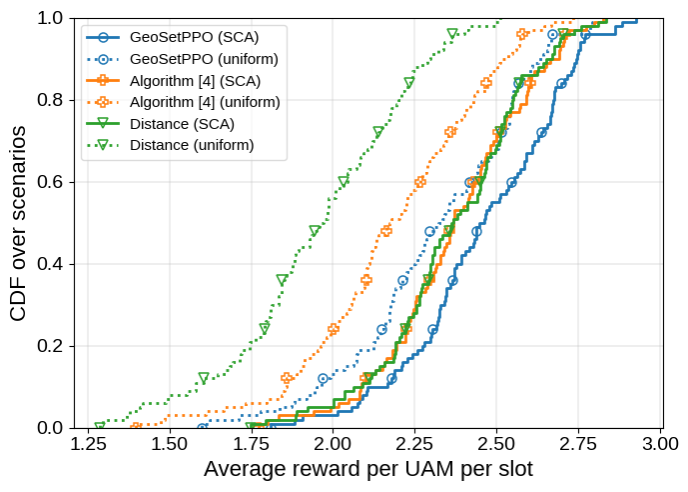}%
			\caption{Empirical CDF of the average reward per UAM over independent test scenarios with $(M,K,B)=(20,4,1)$. Solid and dotted curves represent SCA-based and uniform power allocation, respectively.}
			\label{fig:cdf_b1}
		}
	\end{center}
	\vspace{-10pt}
\end{figure}

Fig.~\ref{fig:sca_convergence} examines the convergence of the inner-loop SCA power allocator. For both $K=4$ and $K=7$, the per-slot reward increases rapidly over a few SCA iterations and then saturates. This indicates that the SCA module provides a fast improvement over uniform powers for a feasible fixed schedule while satisfying the per-GS power and per-link minimum SINR constraints.

The reward alone does not indicate whether the schedule admits a power vector satisfying the minimum SINR constraints. We therefore define the SCA outage probability as the fraction of test slots for which the fixed-schedule problem (P2) is infeasible and report it in Fig.~\ref{fig:sca_outage}. For $B=1$, the algorithm-based scheduler in~\cite{myjsac} yields the lowest outage probability because its association procedure explicitly avoids strong interference, whereas the distance-based scheduler frequently creates interfering links and becomes infeasible even at moderate thresholds. For $B=2$, applying RR subband allocation after association can place mutually interfering links on the same subband without accounting for their angular relationships. GeoSetPPO jointly selects association and subbands and consequently provides lower outage at the higher SINR thresholds. The distance-based method remains the least reliable in all four settings.

Table~\ref{tbl02} reports the measured execution times under the same implementation environment. The RR subband-allocation time for the association baselines is negligible and is therefore not reported separately. GeoSetPPO requires approximately $2.8$--$2.9$~ms per scheduling instance and changes little as the network grows from $(M,K)=(20,4)$ to $(50,7)$. The algorithm-based scheduler increases from $5.545$ to $40.84$~ms because of its combinatorial association procedure, while the distance-based heuristic remains the fastest scheduling method. The SCA power allocator is substantially more expensive than all discrete schedulers, requiring $297.6$~ms for $(20,4)$ and $1871.8$~ms for $(50,7)$ on average. These results distinguish the low inference latency of the learned scheduler from the computational cost of enforcing power and SINR constraints through online optimization.

\subsection{Performance and Robustness Analysis}

\begin{figure}[t]
	\begin{center}
		{\includegraphics[width=0.9\columnwidth,keepaspectratio]
			{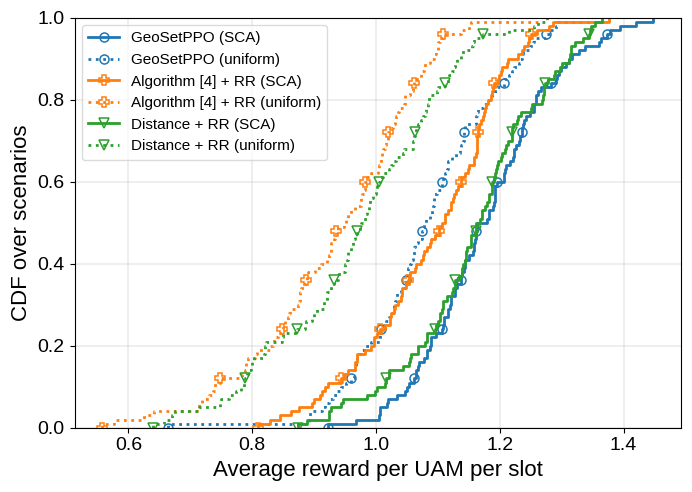}%
			\caption{Empirical CDF of the average reward per UAM over independent test scenarios with $(M,K,B)=(20,4,2)$. RR subband allocation is applied to the association baselines. Solid and dotted curves represent SCA-based and uniform power allocation, respectively.}
			\label{fig:cdf_b2}
		}
	\end{center}
	\vspace{-10pt}
\end{figure}

Fig.~\ref{fig:reward_comparison} compares the average reward for $K=4$ under two subband configurations and two handover-penalty scales. SCA-based power allocation improves the reward relative to uniform power for all schedulers, confirming the benefit of optimizing continuous powers after fixing the discrete schedule. GeoSetPPO achieves the highest or nearly highest average reward across the considered settings, with the clearest margin in the interference-limited $B=1$ case. The model without the final SCA-based fine-tuning stage remains competitive, but does not exploit the SCA-based reward as consistently. This comparison supports the intended role of the second training stage in aligning the scheduling policy with the power allocator used during deployment. For $B=2$, the distance-based method becomes more competitive in raw reward because RR separates part of the interference and short links provide strong desired signals.

Fig.~\ref{fig:cdf_b1} and Fig.~\ref{fig:cdf_b2} report the empirical CDFs of the average reward over independently generated test scenarios for $B=1$ and $B=2$, respectively. For $B=1$, GeoSetPPO with SCA shifts the reward distribution to the right of the baselines over most of the CDF, and the gap between its uniform- and SCA-based curves shows that power optimization provides gains beyond learned association. For $B=2$, the reward distributions of GeoSetPPO and the RR-augmented association baselines overlap substantially, indicating comparable raw reward over part of the test distribution. Because infeasible instances of (P2) use the uniform-power fallback described in Section~\ref{sec_proposed}, these CDFs do not directly characterize schedule feasibility. The reward comparison should therefore be considered jointly with Fig.~\ref{fig:sca_outage}, where GeoSetPPO yields substantially lower outage than the distance-based method and, at higher SINR thresholds, the algorithm-based method with RR.

\begin{figure}[t]
	\begin{center}
		{\includegraphics[width=0.9\columnwidth,keepaspectratio]
			{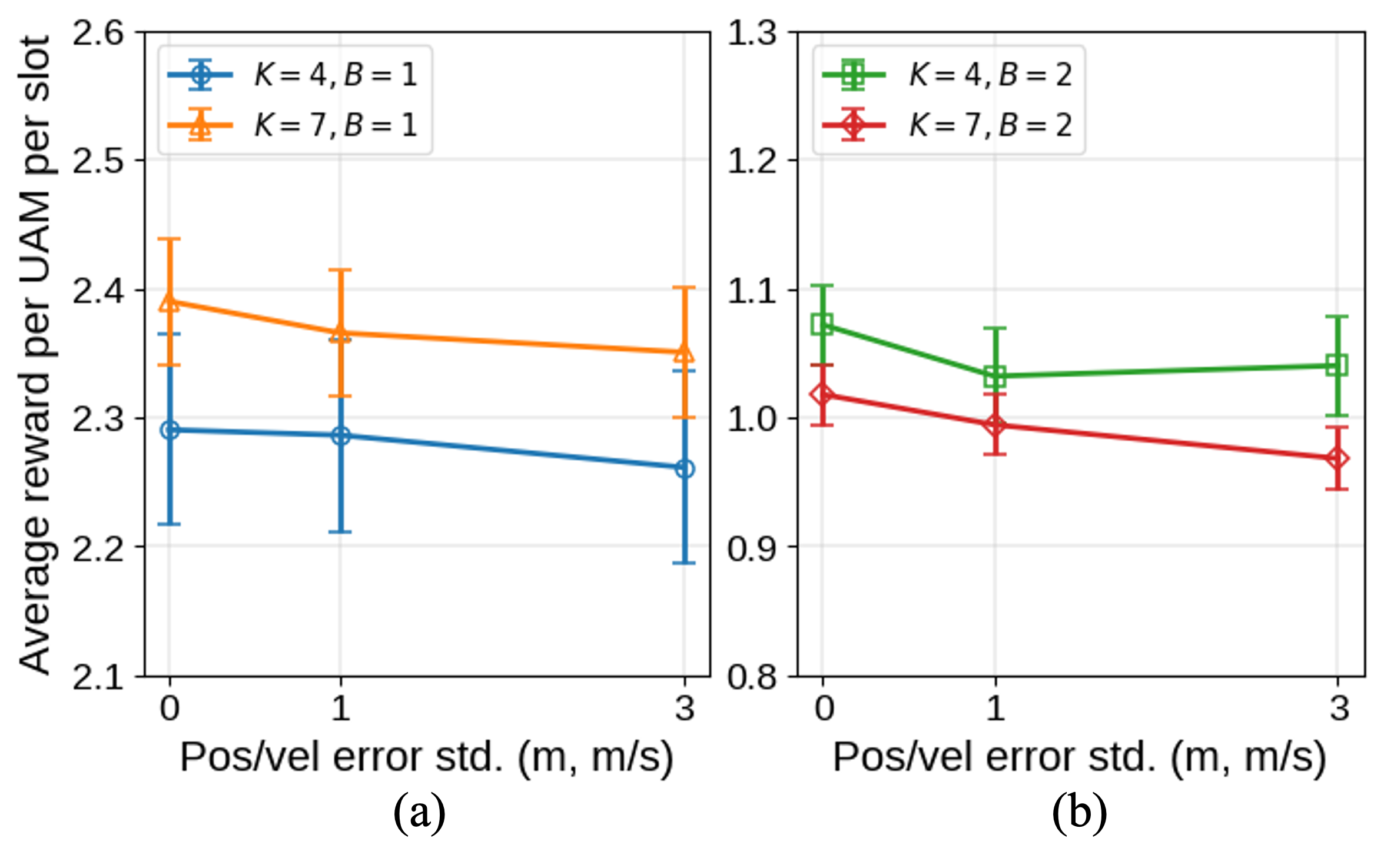}%
			\caption{Average reward of GeoSetPPO under position and velocity input errors. Each horizontal-axis value gives the common per-coordinate standard deviations of the zero-mean Gaussian position error in meters and velocity error in meters per second. Separate models are trained and evaluated under each matched error level.}
			\label{fig:state_noise}
		}
	\end{center}
	\vspace{-10pt}
\end{figure}

\begin{figure}[t]
	\begin{center}
		{\includegraphics[width=0.9\columnwidth,keepaspectratio]
			{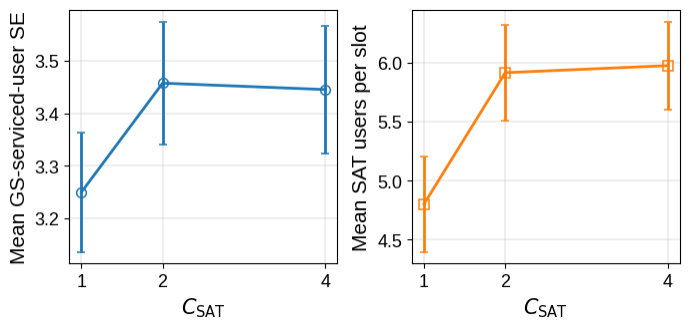}%
			\caption{Sensitivity of the GeoSetPPO scheduling outcome to the satellite-tier sum spectral efficiency $C_\text{SAT}$ for $(M,K,B)=(20,4,1)$. The two panels show the mean spectral efficiency per GS-served UAM and the mean number of satellite-served UAMs per slot, respectively.}
			\label{fig:csat_sensitivity}
		}
	\end{center}
	\vspace{-10pt}
\end{figure}

Fig.~\ref{fig:state_noise} evaluates the effect of imperfect mobility information. Independent zero-mean Gaussian errors are added to each position and velocity coordinate, with the same numerical standard deviation used in meters and meters per second, respectively. The scheduling reward changes only modestly as the error standard deviations increase from $0$ to $3$~m and from $0$ to $3$~m/s across both network sizes and subband settings. The result indicates that GeoSetPPO does not require unrealistically exact mobility information when the expected state-estimation errors are represented during training.

\begin{figure}[t]
	\begin{center}
		{\includegraphics[width=0.9\columnwidth,keepaspectratio]
			{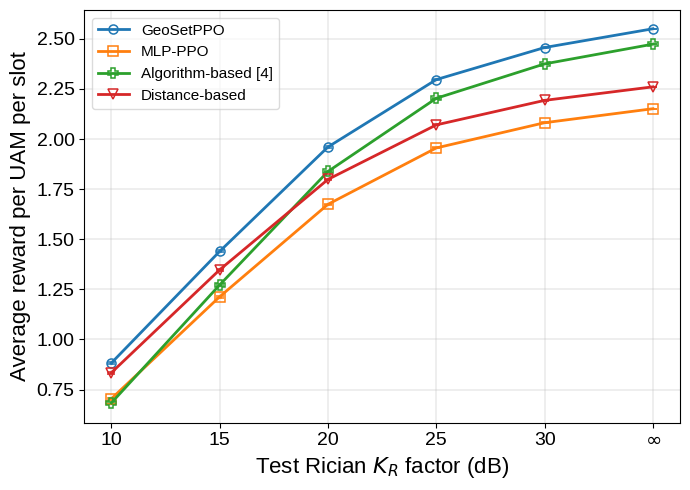}%
			\caption{Average reward per UAM versus the test-channel Rician $K_\mathrm{R}$ factor for $B=1$. Solid and dotted curves represent $(M,K)=(20,4)$ and $(50,7)$, respectively. The DRL models are trained with $K_\mathrm{R}=20$~dB.}
			\label{fig:rician_robustness}
		}
	\end{center}
	\vspace{-10pt}
\end{figure}

Fig.~\ref{fig:csat_sensitivity} compares the scheduling outcomes obtained when a separate GeoSetPPO policy is trained for each value of $C_\text{SAT}$. Increasing $C_\text{SAT}$ from $1$ to $2$ makes satellite association more favorable, increasing the mean number of satellite-served UAMs and reducing interference among the remaining GS-served UAMs. Consequently, the mean spectral efficiency per GS-served UAM also increases. Both quantities change only marginally when $C_\text{SAT}$ is further increased from $2$ to $4$.

Throughout the model development, GeoSetPPO uses geometry-based channel information rather than instantaneous fading realizations. To evaluate the resulting mismatch, Fig.~\ref{fig:rician_robustness} tests models trained at $K_\mathrm{R}=20$~dB over channels ranging from $K_\mathrm{R}=10$~dB to the pure-LoS case. The average reward decreases for all scheduling methods as the Rician factor decreases. Importantly, the baselines exhibit a similar decline, while GeoSetPPO preserves its performance ordering across the tested channel conditions and both network sizes. This common trend indicates that much of the degradation is caused by the reduction in achievable link quality under stronger NLoS components rather than by a failure of the learned scheduler outside its training condition.

\begin{figure*}[t]
	\begin{center}
		{\includegraphics[width=2.0\columnwidth,keepaspectratio]
			{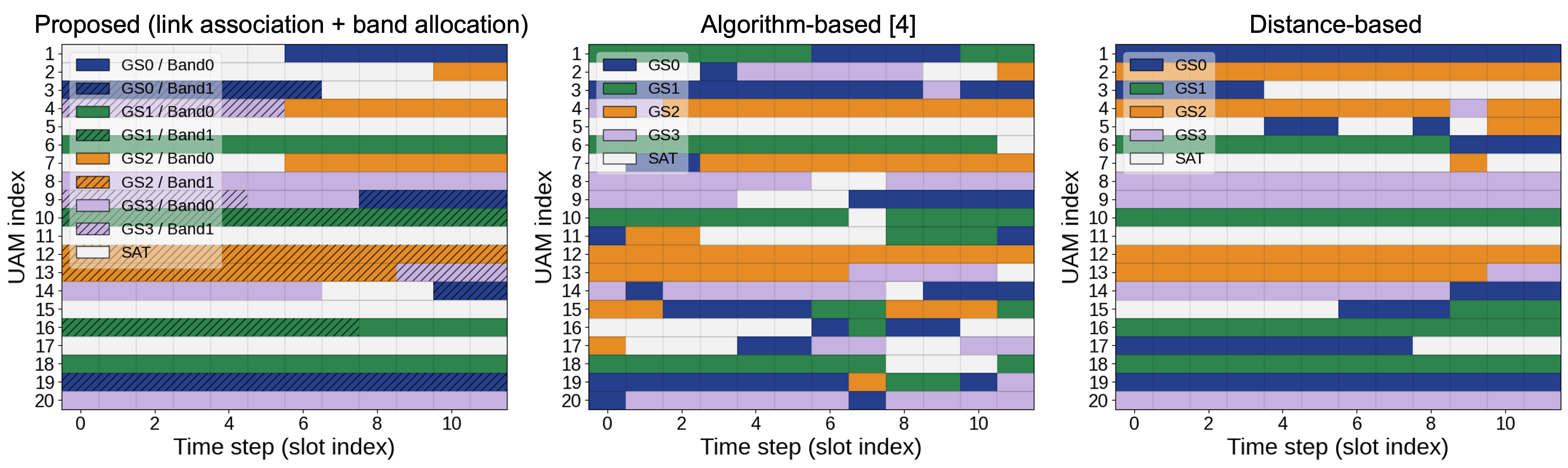}%
			\caption{Example schedules over 12 time steps in a random scenario with $(M,K)=(20,4)$. Left shows the proposed joint association and subband scheduling for $B=2$, where the second subband is indicated by hatch patterns.}
			\label{fig:schedule_example}
		}
	\end{center}
	\vspace{-13pt}
\end{figure*}

Finally, Fig.~\ref{fig:schedule_example} visualizes example schedules for $(M,K)=(20,4)$. GeoSetPPO yields more persistent link patterns over time and avoids unnecessary changes in association and subband usage. This scheduling stability is a principal benefit of the proposed finite-horizon formulation, which explicitly penalizes handovers and resource changes rather than optimizing only the instantaneous sum rate. Combined with sequential optimization over the considered mobility horizon, the use of velocity information allows the policy to anticipate near-future geometry and avoid short-lived switches that provide only transient rate gains. The algorithm-based scheduler exhibits frequent switching across slots, consistent with its snapshot-based decision structure. The distance-based method shows fewer changes due to the temporal correlation of distances, but it still lacks explicit interference awareness and does not optimize subband allocation.

\section{Conclusion}
\label{sec_conc}

This work studied mobility-aware downlink scheduling for UAM communications in a cooperative GS-satellite network with a multi-subband GS tier and geometry-structured interference. In contrast to the single-band formulation in~\cite{myjsac}, which maintains one association over a prediction interval, the proposed formulation permits sequential association and subband decisions while explicitly accounting for handovers over the scheduling horizon. We developed GeoSetPPO, which represents UAM-to-UAM and UAM-to-resource relations through geometry-aware set attention, and combined the learned discrete scheduler with an SCA-based GS power allocator under per-GS power and minimum-SINR constraints. The results demonstrated the importance of the proposed geometric inductive bias relative to conventional MLP- and Transformer-based PPO architectures. GeoSetPPO also achieved favorable reward performance relative to the association baselines while producing more feasible multi-subband schedules at moderate-to-high SINR thresholds. Future work will investigate efficient policy transfer or adaptation across different UAM counts, GS layouts, mobility distributions, subband configurations, and operating objectives.

\appendices

\section{Proof of Lemma~\ref{lemma1}}
\label{appen1}

We use the following logarithmic lower bound~\cite{myfsotraj}:
\begin{equation}
\label{label}
\begin{aligned}
\log(1+x)\geq \theta \log x + \kappa,
\end{aligned}
\end{equation}
where $\theta = \frac{x_0}{1+x_0}$ and $\kappa=\log(1+x_0)-\theta \log x_0$. Applying this approximation to $C_{kbp} = \log\big(1+\gamma_{kbp}\big)$ gives a tight lower bound at the point $x_0=\gamma_{kbp}(\boldsymbol{\rho}^{(\delta)})$, where $\gamma_{kbp}(\boldsymbol{\rho}^{(\delta)})$ can be computed using~\eqref{sinr0}. Thus, for each $kbp$, $\theta_{kbp}$ is given by~\eqref{thetat}, and $\kappa_{kbp}$ is
\begin{equation}
\label{label}
\begin{aligned}
\kappa_{kbp}^{(\delta)} = \log \bigg(1 + \frac{w_{kbp} \rho_{kbp}^{(\delta)}}{\mu_{kbp}^{(\delta)}}\bigg)-\theta_{kbp}^{(\delta)} \log\bigg(\frac{w_{kbp} \rho_{kbp}^{(\delta)}}{\mu_{kbp}^{(\delta)}}\bigg),
\end{aligned}
\end{equation}
where $\mu_{kbp}^{(\delta)}$ is defined in~\eqref{mut}. Using $\theta_{kbp}^{(\delta)}$, $\mu_{kbp}^{(\delta)}$, and $\kappa_{kbp}^{(\delta)}$, we obtain
\begin{equation}
\label{firstbound}
\begin{aligned}
&\log\bigg(1+\frac{w_{kbp}\rho_{kbp}}{\mu_{kbp}}\bigg)\\
&\geq \theta_{kbp}^{(\delta)}\{ \log(w_{kbp} \rho_{kbp}) - \log \mu_{kbp}\} + \kappa_{kbp}^{(\delta)},
\end{aligned}
\end{equation}
where $\mu_{kbp}=\sum_{\ell\neq k}^{K}\sum_{q=1}^{P_{\ell b}} w_{kbp}^{\ell bq} \rho_{\ell bq} + \sum_{q\neq p}^{P_{kb}} w_{kbp}^{kbq} \rho_{kbq} + \sigma_{\mathrm{n},B}^{2}$. Next, we apply the first-order Taylor upper bound of the concave function $\log \mu_{kbp}$ at $\mu_{kbp}^{(\delta)}$:
\begin{equation}
\label{logtaylor}
\begin{aligned}
&\log \mu_{kbp} \leq \log \mu_{kbp}^{(\delta)}+\frac{1}{\mu_{kbp}^{(\delta)}}\bigg\{\sum_{\ell\neq k}^{K} \sum_{q=1}^{P_{\ell b}} w_{kbp}^{\ell bq} (\rho_{\ell bq} - \rho_{\ell bq}^{(\delta)})\\
&\qquad\qquad+\sum_{q\neq p}^{P_{kb}} w_{kbp}^{kbq} (\rho_{kbq}-\rho_{kbq}^{(\delta)}) \bigg\}.
\end{aligned}
\end{equation}
Substituting~\eqref{logtaylor} into~\eqref{firstbound} yields
\begin{equation}
\label{secondbound}
\begin{aligned}
&C_{kbp}
\geq \theta_{kbp}^{(\delta)} \log(w_{kbp} \rho_{kbp})\\
&\qquad\quad-\frac{\theta_{kbp}^{(\delta)}}{\mu_{kbp}^{(\delta)}}\bigg(\sum_{\ell\neq k}^{K} \sum_{q=1}^{P_{\ell b}} w_{kbp}^{\ell bq}\rho_{\ell bq}+\sum_{q\neq p}^{P_{kb}} w_{kbp}^{kbq}\rho_{kbq}\bigg)+\zeta_{kbp}^{(\delta)},
\end{aligned}
\end{equation}
where $\zeta_{kbp}^{(\delta)}$ is defined in~\eqref{zetat}. Taking the negative of both sides gives the convex upper bound for $-C_{kbp}$ in~\eqref{approx01}.

\bibliographystyle{IEEEtran}
\bibliography{HJM_ref}

\end{document}